\documentclass[sn-mathphys-ay]{sn-jnl}

\usepackage{graphicx}%
\usepackage{adjustbox}
\usepackage{multirow}%
\usepackage{ulem}%
\usepackage{amsmath,amssymb,amsfonts}%
\usepackage{amsthm}%
\usepackage{mathrsfs}%
\usepackage[title]{appendix}%
\usepackage{xcolor}%
\usepackage{textcomp}%
\usepackage{manyfoot}%
\usepackage{booktabs}%
\usepackage{algorithm}%
\usepackage{algorithmicx}%
\usepackage{algpseudocode}%
\usepackage{listings}%
\usepackage{subcaption}
\newtheorem{theorem}{Theorem}[section]
\newtheorem{lemma}[theorem]{Lemma}
\usepackage{placeins}

\usepackage[T1]{fontenc}
\usepackage[font=small,labelfont=bf,tableposition=top]{caption}
\DeclareCaptionLabelFormat{andtable}{#1~#2  \&  \tablename~\thetable}
\usepackage{amsmath}

\theoremstyle{thmstyleone}%
\theoremstyle{thmstyletwo}%

\theoremstyle{thmstylethree}%
\newtheorem{definition}{Definition}%

\begin{document}

\title[Article Title]{Independent Approximates provide a maximum likelihood estimate for heavy-tailed distributions}


\author*[1,]{\fnm{Amenah} \sur{Al-Najafi}}\email{amenah.alnajafi@gmail.com}

\author[2]{\fnm{Ugur} \sur{Tirnakli}}\email{ugur.tirnakli@ieu.edu.tr}
\equalcont{These authors contributed equally to this work.}

\author[3]{\fnm{Kenric} \sur{P. Nelson}}\email{kenric.nelson@photrek.io}
\equalcont{These authors contributed equally to this work.}

\affil*[1]{\orgdiv{Department}, \orgname{Organization}, \orgaddress{\street{Street}, \city{City}, \postcode{100190}, \state{State}, \country{Country}}}

\affil[1]{\orgdiv{Department of Mathematics}, \orgname{University of Kufa}, \orgaddress{\street{299G}, \city{Najaf}, \country{Iraq}}}

\affil[2]{\orgdiv{Department of Physics, Faculty of Arts and Sciences}, \orgname{Izmir University of Economics}, \orgaddress{\street{35330}, \city{Izmir}, \country{Turkey}}}

\affil[3]{\orgname{Photrek, LLC}, \orgaddress{\street{56 Burnham St Unit 1}, \city{Watertown, MA}, \country{USA}}}


\abstract{

Heavy-tailed distributions are infamously difficult to estimate because their moments tend to infinity as the shape of the tail decay increases. Nevertheless, this study shows that a modified group of moments can be used to determine a maximum likelihood estimate of heavy-tailed distributions. These modified moments are determined from powers of the original distribution. Within nonextensive statistical mechanics, this has been referred to as the escort distribution. Here we clarify that this is the distribution of Independent-Equals, the independent random variables sharing the same state. The $n$th-power distribution is guaranteed to have finite moments up to $n-1$. Samples from the $n$th-power distribution are drawn from $n$-tuple Independent Approximates, which are the set of independent samples grouped into n-tuples and sub-selected to be approximately equal to each other. We show that Independent Approximates are a maximum likelihood estimator for the generalized Pareto and the Student’s t distributions, which are members of the family of coupled exponential distributions. We use the first (original), second, and third power distributions to estimate their zeroth (geometric mean), first, and second power-moments respectively. In turn, these power-moments are used to estimate the scale and shape of the distributions. A least absolute deviation criteria is used to select the optimal set of Independent Approximates. Estimates using higher powers and moments are possible though the number of $n$-tuples that are approximately equal may be limited.}




\maketitle

\section{Introduction}\label{sec1}

Heavy-tailed distributions play a crucial role in characterizing the statistics of complex systems and find applications in various fields, including finance, climatology, telecommunication, genetics, and more \cite{resnickHeavyTailPhenomenaProbabilistic2007}, \cite{merzUnderstandingHeavyTails2022}, \cite{ibragimovHeavyTailedDistributionsRobustness2015} and \cite{bradleyFinancialRiskHeavy2003}. Additionally, the Student's t distribution has been employed to model logarithmic returns in simulated markets \cite{granhaOpinionDynamicsFinancial2022}. These distributions have tails that decay slower than the exponential, producing outlier samples that cause some moments to diverge, which makes their estimation challenging. Various tail index estimators exist, but they are sensitive to the choice of the upper order statistics parameter $( k )$, making optimal selection difficult.

Nonextensive statistical mechanics (NSM) 
\cite{tsallisFoundationsStatisticalMechanics2017,abeNonextensiveStatisticalMechanics2001a} has been proposed to model the statistics of complex systems; however the physical property of its generalizing index \(q\) is typically undefined.  In this paper, we clarify that the escort probability, $P^{(q)}\equiv \frac{p_i^q}{\sum_jp_j^q}$, is the distribution of \textit{q} random variables sharing the same state, which we refer to as the distribution of \textit{Independent} \textit{Equals.} The estimation method based on this property, Independent Approximate (IA) estimation method \cite{nelsonIndependentApproximatesEnable2022}, is optimized to achieve a maximum likelihood estimation. In this paper, we focus a filtering method that produces independent approximates with an integer values of $q$, particularly $P^{(2)}$ and $P^{(3)}$ but this motivates further investigation to define a sampling method for fractional independent approximates. Utilizing the least absolute deviations (LAD), we obtain the optimal subsample set for statistical analysis. We assess estimator performance using metrics such as the Coefficient of Efficiency, Average Deviation (AD), Cramér–von Mises (CvM), and Negative Log-Likelihood (NLL).

While the NSM methodologies have focused on the independent-equals $q$ as a defining property, we suggest that refocusing on the degree of nonlinearity $\kappa$ provides a framework more closely grounded in the fundamental property of a complex system. This approach was previously introduced by one of the authors (Nelson) as the nonlinear statistical coupling (NSC) \cite{nelsonNonlinearStatisticalCoupling2010, nelsonAverageUncertaintySystems2017}. The NSC or coupling $\kappa$ is equal to the the shape parameter of the generalized Pareto distribution (GPD), and the inverse of the degree of freedom of the Student's t distribution. Furthermore the coupling is either equal to or closely related to the nonlinearity of stochastic systems which deviate from exponential decaying noise \cite{beckDynamicalFoundationsNonextensive2001, nelsonOpenProblemsNonextensive2024}.  

By reframing NSM around the nonlinearity of a complex system, the methods can be grounded in longer-standing traditions of statistics (such as the GPD and Student's t) and the central issues regarding characterization of complex systems. Section 2 provides background on the coupled distribution family, nonlinear sources of heavy-tailed distributions, and estimation methods. In Section 3, the IA estimation method is defined. In Section 4, the performance of the IA method as a MLE and its application to estimation of several nonlinear simulations are documented. Section 5, provides a discussion and conclusion regarding the relationship between nonlinearity and independent-equals, and how this relationship enables a maximum likelihood estimation method using independent approximates

\section{Background}\label{background}
\subsection{The Coupled Distribution Family}

We will examine the performance of the IA estimator for the generalized Pareto and the Student's t distributions. To facilitate the discussion, we begin by recalling the broader family of \textit{q}-distributions introduced in the context of nonextensive statistical mechanics.
The family of Tsallis \textit{q}-probability density functions (PDF) is defined as:
\[
f(x;q, \beta, \alpha)=\frac{1}{Z_q\left( \beta,\alpha \right)} \left[1+(1-q)\beta x^{\alpha}\right]^{1 /(1-q)},
\]
where $(\alpha=1)$ corresponds to the \textit{q}-exponential distribution and $(\alpha=2)$ to the \textit{q}-Gaussian distribution.

Since we wish to model the statistics of complex systems with its foundational property, the nonlinearity, rather than the number of independent-equals, we reparameterize this family in two steps. We'll show that this is achieved by starting with the survival function for the GPD. The shape parameter, $\kappa$, defines the nonlinearity of the generalized logarithm and the rate of decay of the generalized exponential:
$$
\ln_\kappa{x}\equiv\frac{1}{\kappa}\left(x^\kappa-1\right); \quad x>0,\quad \exp_\kappa{x}\equiv\left(1+\kappa x\right)_+^{\frac{1}{\kappa}}, \quad \left(a\right)_+=\max\left(0,a\right).
$$
Given the connection to the nonlinearity, which we'll expand upon in the next subsection, the shape parameter is also called the nonlinear statistical coupling or coupling for short. The reciprocal of the generalized exponential function is the survival function (SF) (1 - CDF) of generalized Pareto distribution:
$$
S\left(x;\mu,\sigma,\kappa\right)=1-F\left(x;\mu,\sigma,\kappa\right)\equiv\left(\exp_\kappa{\left(\frac{x-\mu}{\sigma}\right)}\right)^{-1}.
$$
The coupled distribution is defined such that the power $\alpha$ of the variable $x$ is one for the generalized Pareto distribution and two for the Student's t distribution. The symmetric two-sided version of the coupled CDF, like the Student's t CDF, is a function of the regularized incomplete beta function, $I(z;a,b)$  
\begin{equation}\label{CD-CDF}F\left(x;\mu,\sigma,\kappa,\alpha\right)\equiv\left\{\begin{matrix}\frac{1}{2}I\left(\frac{1}{1+\kappa\left|\frac{x-\mu}{\sigma}\right|^\alpha},\frac{1}{\alpha\kappa},\frac{1}{\alpha}\right)&x\le\mu\\
\frac{1}{2} + \frac{1}{2} I\left(\frac{\kappa\left|\frac{x-\mu}{\sigma}\right|^\alpha}{1+\kappa\left|\frac{x-\mu}{\sigma}\right|^\alpha},\frac{1}{\alpha},\frac{1}{\alpha\kappa}\right)& \text{True}\\ \end{matrix}\right.\end{equation}
The two-sided coupled PDF is defined as
\begin{equation}\label{CD-PDF}f\left(x;\mu,\sigma,\kappa,\alpha\right)\equiv\left\{\begin{matrix}\frac{1}{Z\left(\sigma,\kappa,\alpha\right)}\left(1+\kappa\left|\frac{x-\mu}{\sigma}\right|^\alpha\right)_+^{-\frac{1+\kappa}{\alpha\kappa}}&\kappa\neq0\\\frac{1}{Z\left(\sigma,\alpha\right)}\exp\left(-\frac{1}{\alpha}\left|\frac{x-\mu}{\sigma}\right|^\alpha\right)&\kappa=0\\\end{matrix}\right.,\sigma\geq0,-1\le\kappa,0<\alpha<2\end{equation}
where the generalized Pareto (coupled exponential) distribution is $\alpha = 1$ and the Student's t (coupled Gaussian) distribution is $\alpha =2$. The coupling is the inverse of the degree of freedom for the Student's t. Although we focus on the univariate case in this paper, the coupled distribution can be extended to the multivariate setting by replacing the univariate deviation term with a generalized Mahalanobis form, $((x - \mu)^T \Sigma^{-1} (x - \mu))^{\alpha/2}$. This formulation enables flexible control of tail decay via the parameter $\alpha$, generalizing the standard Mahalanobis distance ($\alpha = 2$) to a broader class of heavy-tailed multivariate distributions. For the multivariate distribution the exponent is $-\frac{1+d\kappa}{\alpha\kappa}$, where $d$ is the number of dimensions.

The normalization or partition function, $Z$, is
\begin{equation}\label{Z}Z\left(\sigma,\kappa,\alpha\right)\equiv\left\{\begin{matrix}\sigma&\alpha=1\\\begin{matrix}\frac{\sigma\ B\left(\frac{1}{2\kappa},\ \frac{1}{2}\right)}{\sqrt\kappa}&\kappa>0\\\sigma\sqrt{2\pi}&\kappa=0\\\frac{\sigma\ B\left(\frac{-1+\kappa}{2\kappa},\ \frac{1}{2}\right)}{\sqrt{-\kappa}}&-1\le\kappa<0\\\end{matrix}&\alpha=2\\\end{matrix}\right.\end{equation}
where $B$ is the (incomplete, $z\neq1$) Beta function, $B\left(z,a,b\right)=\int_{0}^{z}{t^{a-1}\left(1-t\right)^{b-1}dt}$ and the regularized incomplete beta function is $I(z;a,b)=\frac{B(z;a,b)}{B(a,b)}$. $\sigma$ is the scale parameter, equal to the standard deviation when $\kappa=0$, and corresponds to the "knee" in the log-log plot of the PDF. On the log-log plot the PDF has zero slope at $x=\mu$, and a slope of $-\frac{1+\kappa}{\kappa}$ for $x\rightarrow\infty$. At $x=\sigma$ the log-log slope is -1 for all values of $\kappa$ and $\alpha$. Also, at $x=-\frac{\sigma}{\kappa^{\alpha}}$ the slope is half the value at infinity, $-\frac{1+\kappa}{2\kappa}$ and the second derivative is a maximum.  These properties are reviewed further in  Appendix \ref{secA1}. 

Relating the exponents of the coupled exponential family and the $q$ distributions we recover the following relationships:
\begin{equation}\label{q-coupling}
q=1+\frac{\alpha\kappa}{1+\kappa};\qquad \kappa=\frac{q-1}{\alpha-q+1}.
\end{equation}
From the term multiplying the variable $x$, the relationship between scale parameters $\beta$ and $\sigma$ is:
\begin{equation}
    \beta=\frac{1+\kappa}{\alpha\sigma^\alpha};\qquad \sigma=\left(\beta(\alpha+1-q)\right)^{-\frac{1}{\alpha}}.
\end{equation}

There are three domains of the coupled distributions,  heavy-tailed for $\kappa>0$ $(1<q<1+\alpha$) ; exponential for $\kappa=0$ $(q=1)$; and compact-support for $-1<\kappa<0$ $(0<q<1)$.

\subsection{Nonlinear sources of heavy-tailed distributions}
\begin{lemma}
Let a stochastic process $X_t$ be defined as
\begin{equation*}
    dX_t=f(X_t)dt+AdW_t^{(a)}+g(X_t)MdW_t^{(m)}
\end{equation*}
where $dW_t^{(a)}$ and $dW^{(m)}_t$ are independent Wiener processes which define the additive (a) and multiplicative (m) noise, and A and M are the amplitudes of each noise source. And let $f(x)=-\tau g(x)g'(x)=-V'(x)$, where $V(x)$ is a potential function specifying the influence of the multiplicative noise on the diffusion $D(x)$. Then the probability density $p_X(x,t)$ for this system has a limit solution of $p_X(x)=\lim_{t\to\infty}p_X(x,t)\propto \exp^{-\frac{1+\kappa}{2}}_\kappa \left(\frac{g^2(x)}{\sigma^2}\right)$, with $\sigma^2=\frac{A}{\tau}$ and $\kappa =\frac{D'(x)}{V'(x)}=\frac{M}{\tau}$.
\end{lemma}
\begin{proof}
    The probability density $p_X (x, t)$ is the solution to the Fokker–Planck equation \cite{anteneodoMultiplicativeNoiseNonextensive2003}
\begin{equation} \label{pd}
    \partial p_X (x, t) \partial t = -\frac{\partial}{\partial x} [J(x)p_X (x, t)] + \frac{\partial ^2}{\partial x^2} [D(x)p_X (x, t)] 
\end{equation}where $J(x)$ is the drift and D(x) is the diffusion:
\begin{align*}
    J(x)\equiv f(x)+Mg(x)g'(x)\\
    D(x)\equiv A+\frac{M}{2}[g(x)]^2
.
\end{align*}

Since $g(x)$ is related to the potential function by $V'(x)=\tau g(x)g'(x)\rightarrow V(x)=\frac{\tau}{2}[g(x)]^2$, the diffusion rate and the potential function are related by the expression
\begin{equation}
    D(x)=A+\frac{M}{\tau}V(x).
\end{equation}\label{pd_1}
The magnitude of the nonlinear source is $D'(x)/V'(x)=\frac{M}{\tau}$ of the coupled-exponential family \cite{naudtsGeneralisedExponentialFamilies2008} ,and as shown in \cite{anteneodoMultiplicativeNoiseNonextensive2003}.The stationary solution $p_X(x)$ for \ref{pd} is a member
\begin{align}
    p_X(x)=\lim_{t\to\infty}p_X(x,t)&\propto \left(1+\frac{M}{A}g^2(x)\right)^{-\frac{1+M/\tau}{2M/\tau}}\\
    &\propto\exp_\kappa^{-\frac{1+\kappa}{2}}\left[\frac{g^2(x)}{\sigma^2}\right]
\end{align}
Equating the exponents determines the coupling to be $\kappa=\frac{M}{\tau}$ and separating the coupling from the scale determines $\sigma^2=\frac{A}{\tau}$, completing the proof.
\end{proof}

In section \ref{subsubsec2} we compare estimation methods for data from the Coherent Noise Model (CNM) and the Standard Map Model (SMM). Both of these chaotic processes have additive and multiplicative noise components that result in coupled exponential and coupled Gaussian limit distributions, respectively. The relationship between the process physics and the scale and shape of the limit distribution is an open area of research. For instance, Zubillaga, et. al \cite{zubillagaThreestateOpinionDynamics2025} showed that simulations of financial markets have a nontrivial relationship between the inputs and the output shape and scale of the limit distribution.

\subsection{Estimation methods for heavy-tailed distributions}

The classical estimators for estimating the tail index in heavy-tailed distributions are widely used. The Hill's estimator \cite{hill1975tail} is one of the most common estimators for the tail index of heavy-tailed distributions. It is a type of maximum likelihood estimator that relies on the upper order statistics of the sample. A key limitation of Hill’s estimator lies in its dependence on the threshold parameter $k$, which determines the number of top order statistics used in the estimation. Selecting an appropriate value for $k$ is not straightforward, and the estimate is often very sensitive to this choice. This has motivated the development of alternative tail index estimators applicable to a wider class of heavy-tailed distributions. Dekkers, Einmahl, and de Haan \cite{dekkers1989moment} proposed a generalization known as the moment estimator, which extends the estimation to the broader case where  $\xi\in\mathbb{R}$. This modification incorporates higher-order moments to improve the estimator’s properties, particularly in finite samples. The moment estimator retains the key asymptotic properties of consistency and asymptotic normality, while reducing sensitivity to the threshold parameter $k$ compared to the Hill estimator. The simplest estimator for $\xi$ is the estimator of Pickands \cite{pickands1975statistical}. This estimator does not rely on logarithmic spacing of order statistics. Like Hill's method, Pickands' estimator is based on the upper order statistics of the sample and depends on a threshold parameter k. However, it avoids some of the instability of the Hill estimator in small samples by using ratios of differences between order statistics

Utilizing the NSM framework, Shalizi \cite{shaliziMaximumLikelihoodEstimation2007} established a maximum likelihood estimator for the q-exponential distribution, which is equivalent to the Generalized Pareto Distribution. In addition, Çadırık \cite{cadirci2025nonparametric} proposed a non-parametric goodness-of-fit testing framework for $q$-exponential distributions based on Tsallis entropy. This method assesses whether the underlying data distribution fit with a $q$-exponential model. This makes it a useful complement to estimation approaches, particularly in validating model assumptions before applying inference techniques.

\section{Estimation using Independent Approximates}\label{sec3}
In this paper, we analyze the performance of estimation using Independent Approximates \cite{nelsonIndependentApproximatesEnable2022}, in which independent samples from a random variable $X\sim f(x)$ are split into n-dimensions for subsampling of n-tuplets that are approximately equal. Our theoretical analysis assumes samples that are restricted to the equality condition $x_1 = x_2 = \cdots = x_n$ from the joint distribution of $n~i.i.d.$ copies of $X$. These "Independent-Equals" follow the conditional distribution:
\[
X^{(n)} \sim f^{(n)}(x) = \frac{f(x)^n}{\int_{x\in X} f(x)^n dx} = f(x_1, \dots, x_n \mid x_1 = \cdots = x_n)
\]
Our experimental approach entails selecting samples from a tolerance neighborhood $\epsilon$ around the independent equals conditional distribution, which we term Independent Approximates (IAs), denoted as $X^{(n)} \sim \hat{f}^{(n)}(x)$.
The idea behind selecting these Independent Approximate samples is to obtain a subset of the data that has a lower shape parameter and thus has moments that can be estimated, while still retaining a functional relationship with the original distribution's parameters.

In the case of the coupled distribution family, we specifically focus on the IA method when the shape parameter $\alpha$ is 1 or 2. We create estimators for the nth power of a density by taking $N$ independent and identically distributed samples and dividing them into subsets of length $n$. The subsets that have absolute values approximately equal are then selected. The median of these selected subsets gives us a set of Independent Approximates (IA) of size $N^{(n)}\equiv \lfloor(N/n)\rfloor$.

\begin{definition}[Power-moment]
The $mth$ moment of the $nth$ power of density $f_X (x)$ function defined as
$$\mu^{(n)}_m=\dfrac{\int_{x\in X}x^mf^n_X(x)dx}{\int_{x\in X}f^n_X(x)dx}=E\big[\big(X^{(n)}\big)^m].$$
\end{definition}

To illustrate the concept of the $m$th moment of the $n$th power density, let's consider the zeroth moment.In this scenario, we can define $f^{(n)}(x)$ as 
$$f^{(n)}=\frac{f^n(x)}{\int_{x\in X}f^n(x)}$$
The above expression represents the density of the conditional distribution along the diagonal of $n$ equal values when the $m$th moment is zero.

The concept of power-density moments allows for a mapping between estimates that can be obtained from the reduced shape of a distribution and those of the original distribution. If a distribution is raised to the power of $(m + 1)$ and renormalized, then the $m$th moment exists and is finite for all shapes $(\kappa \geq 0)$, \cite{nelsonIndependentApproximatesEnable2022}. Table \ref{GPD} provides the functional relationship between moments and the nth power of the generalized Pareto distribution, considering the location, scale, and shape parameters.

In the subsequent sections, we delve into the application and analysis of the IA method for the coupled distribution family, with a specific focus on the cases where $\alpha$ is 1 or 2.

\subsection{Selecting independent approximate subsamples (IAs)}
Our approach builds upon Nelson's IA method \cite{nelsonIndependentApproximatesEnable2022}, introducing a more general class of estimators and applying the least absolute deviation (LAD) to carefully select the optimal subsample. The algorithm is detailed below, providing a comprehensive understanding of our approach.
Use of the LAD enhances the robustness and applicability of the IA method.
Let $X_1, X_2, \dots, X_N$ be independent, identically distributed samples from a one-dimensional random variable with a one dimensional distribution $X_i \sim F(x)$. We randomly partition the samples into $n-tuplet$ groups, so there are $I=\lfloor N/n\rfloor$ groups denoted by $(X_{i_1}, X_{i_2},\dots X_{i_n}),(X_{i_{n+1}}, X_{i_{n+2}},\dots X_{i_{2n}}),\dots (X_{i_{(n*I)-(n-1)}}, X_{i_{(n*I)-(n-2)}},\dots X_{i_N})$, where $i_1, i_2,\dots,i_N$ is a random permutation of the integers $1, 2, \dots, N$. Let $\mathcal{D}$ be the set of all absolute differences calculated for each group as $d_n = \max(X_{i_{j}}^{(n)}) - min(X_{i_{j+1}}^{(n)})$. Since the samples $X_1, X_2, \dots, X_N$ are independently and identically distributed, the elements in $\mathcal{D}$ will be randomly distributed. Now, we sort the elements in $\mathcal{D}$ in ascending order, denoted by $\mathcal{D}_{\text{sorted}}$, such that $d_1 \leq d_2 \leq \dots \leq d_N$.Thus the Independent Approximate is
\[IAs=:\{d_i\in \mathcal{D}_{\text{sorted}}|d_i\leq\varepsilon\}\]
Then, we select the median $M_n$ of each $I$-tuple as the subsample,i.e., $M_n=\text{Median}(X_{i_1}, X_{i_2},\dots,X_{i_n})$ for $(X_{i_1}, X_{i_2},\dots,X_{i_n})\in IAs$. In this study, the threshold $\epsilon$ is not predefined but is instead determined dynamically by the algorithm during the data processing. The algorithm starts by considering pairs with the smallest absolute differences and gradually increases the threshold until it finds the appropriate limit that yields the best result. The selection of the optimal threshold relies on minimizing the least absolute deviation. This ordering of the absolute deviation is similar to the order statistics analysis used in the Hill estimator for the shape.

In order to estimate the parameters of the heavy tail distribution, we consider $X_1, X_2,\cdots, X_N$ independent identically distribution $(iid)$ drawn from a coupled distribution, with $\alpha$ either 1 or 2, where $N$ denotes the sample size. 
We use 25 permutations to select the independent approximate samples (IAs). Under the assumption that $\mu=0$, the (IAs) and the optimal number of sub-samples are selected as: 

\begin{enumerate}
\item If estimating scale:
\begin{enumerate}
\item  Partition the samples randomly into pairs and select equal pairs using $d_k=\{(X^{(2)}_{i_{2I-1}}-X^{(2)}_{i_{2I}}):~I=1,2,\dots,\lfloor N/2\rfloor, \lvert X^{(2)}_{i_{2I-1}}-X^{(2)}_{i_{2I}}\rvert$\}, sort the pairs by their distance.\label{a}
\item Find the median of these pairs $M_I$ to form a vector of independent approximation (IA) $M_I=Median\{X_{i_{2I-1}}, X_{i_{2I}}\}~~\text{for}~(X_{i_{2I-1}}, X_{i_{2I}}\in d_I)$.\label{b}

\item To determine the optimal number of sub-samples for  estimation, follow these steps:\label{c}

\begin{enumerate}
\item Choose a sample size from \ref{b}.In our experiments we started with either 10 or 20 and increased by one each cycle.\label{i}
        \item Estimate the $\sigma$ using the first table where $\mu=0$. 
        \item Check if the estimator has the smallest standard deviation, if yes, choose it and stop. If not, go back to Step \ref{a} and choose another sample size.\label{ii}
        \item Repeat steps \ref{i}-\ref{ii} until the estimator with the smallest standard deviation is found.
\end{enumerate}
\end{enumerate}
\item In case estimating the $\kappa$: 
\begin{enumerate}
\item We have expanded the procedure for the Independent Approximate Subsamples (IAs) from pairs to triplets. The samples are randomly partitioned into $\lfloor N/3\rfloor$ triplets, represented as $(X_{i_1}, X_{i_2}, X_{i_3}), (X_{i_4}, X_{i_5}, X_{i_6}), \dots, (X_{i_{N-2}}, X_{i_{N-1}}, X_{i_N})$, where $i_1, i_2, \dots, i_N$ is a random permutation of the integers $1, 2, \dots, N$. For each triplet, we calculate the absolute differences $d_I = \max(|X_{i_{3I-2}} - X_{i_{3I-1}}|, |X_{i_{3I-1}} - X_{i_{3I}}|)$ and retain only those triplets whose maximum difference is less than or equal to the tolerance $\epsilon$.The optimal number of sub-samples can be found using the same step \ref{c}. 
\end{enumerate}
\end{enumerate}

The finite $n-1$ moment occurs when the density function $f$ is raised to the power of $n$, ensuring that the bias remains finite. In the subsequent analysis, we assume that the subsamples to be independent and identical. To simplify notation, we assume $\mu$ as a known constant. Therefore, in proving the asymptotic results for $\sigma$ and $\kappa$, we will utilize the first and second moments, respectively. Let $\mu^{(n)}_m$ represent the general $m$th moment for a $n$th power density function. The expressions for the general power moments of the Generalized Pareto Distribution (GPD) are as follows:

\begin{align}
    \mu^{(n)}_m=\left(\dfrac{\kappa}{\sigma}\right)^{-m}\, \dfrac{m!\left(-2-m+n+\dfrac{n}{\kappa}\right)!}{\left(-2+n+\dfrac{n}{\kappa}\right)!}~~\text{for}~\kappa<\dfrac{n}{2+m-n}.\label{general_power}
\end{align}
Finally, we define $\hat{\sigma}$ as the estimate of the scale parameter and $\hat{\kappa}$ as the estimate of the shape parameter, given by $\hat{\sigma}=2\hat{\mu}^{(2)}$ and $\hat{\kappa}=\frac{8(\mu_1^2)^{(2)}}{3\mu_1^{(3)}}-3$ for coupled exponential and $\hat{\sigma}=\sqrt{3\mu^{(3)}_2}$ for coupled Gaussian. $I^2$ and $I^3$ represent pairs and triplets of independent approximations, respectively. These moment estimates are based on the values presented in \cite[Table 1]{nelsonIndependentApproximatesEnable2022} and \cite[Table 3]{nelsonIndependentApproximatesEnable2022}. For convenience, we list these tables here, Tables \ref{GPD}, \ref{t_dis}.

In this paper we focus on the estimation of the scale and shape assuming the location is known.  For the coupled Gaussian estimating three parameters is straightforward since $\mu_1^{(2)}$ and $\mu_2^{(3)}$ are independent of each other.  For the coupled exponential distribution, further investigation is required to determine a set of three moments that are sufficiently independent.

\begin{table}[!ht]
\caption{The $m$th moments for the $(m+1)$-power-density of the one-sided GPD.}\label{GPD}
\begin{tabular}{l|ll}
\hline
\multirow{2}{*}{\begin{tabular}[c]{@{}l@{}}Moment,\\ Centered\end{tabular}} & \multicolumn{2}{c}{One side Pareto type II} \\ \cmidrule{2-3}
 & \multicolumn{1}{l|}{Non centered} & Centered\\ \hline
$\mu_0, x-\mu$ & \multicolumn{1}{c|}{-} & (Geometric mean or Log-Average)\\ \hline
$\mu_1^{(2)}, x-\mu$ & \multicolumn{1}{c|}{$\mu+\dfrac{\sigma}{2}$} & \multicolumn{1}{c}{$\dfrac{\sigma}{2}$} \\ \hline
$\mu_2^{(3)}, x-\mu$ & \multicolumn{1}{c|}{$\mu^2+\dfrac{2\mu\sigma}{3+\kappa}+\dfrac{2\sigma^2}{3(3+\kappa)}$} & \multicolumn{1}{c}{$\dfrac{2\sigma^2}{3(3+\kappa)}$} \\ \hline
$\mu_3^{(4)}, x-\mu$ & \multicolumn{1}{c|}{$\mu^3+\dfrac{3\mu^2\sigma}{(4+2\kappa)}+\dfrac{12\mu\sigma^2+3\sigma^3}{2(4+\kappa)(4+2\kappa)}$} & \multicolumn{1}{c}{$\dfrac{3\sigma^3}{2(4+\kappa)(4+2\kappa)}$} \\ \hline
$\mu_4^{(5)}, x-\mu$ & \multicolumn{1}{c|}{$\mu^4+\dfrac{4\mu^3\sigma}{5+3\kappa}+\dfrac{12\mu^2\sigma^2}{(5+2\kappa)(5+3\kappa)}+\dfrac{24\sigma^4}{5(5+\kappa)(5+2\kappa)(5+3\kappa)}$} & \multicolumn{1}{c}{$\dfrac{120\mu\sigma^3+24\sigma^4}{5(5+\kappa)(5+2\kappa)(5+3\kappa)}$} \\ \hline
$\mu_{(m)}^{(m+1)}, x-\mu$ & \multicolumn{1}{c|}{$\sum_{i=0}^{m}\dfrac{m!\mu^{m-i}\sigma^i}{(m-1)!}\dfrac{\dfrac{1+m+(m-i-1)\kappa}{\kappa}!}{(\kappa^i)\dfrac{1+m+(m-1)\kappa}{\kappa}!}$} & \multicolumn{1}{c}{$\kappa^{-m}\dfrac{m!\dfrac{1+m-\kappa}{\kappa}!}{\dfrac{1+m+m\kappa-\kappa}{\kappa}!}\sigma^{-m}$} \\ \hline
\end{tabular}
\end{table}

\begin{table}[!ht]
\caption{The $nth$ moments for the $(n+1)$-power-density the Student's t-distribution.}\label{t_dis}
\begin{tabular}{l|ll}
\hline
\multirow{2}{*}{\begin{tabular}[c]{@{}l@{}}Moment,\\ Centered\end{tabular}} & \multicolumn{2}{c}{Student’s t, n+1} \\ \cmidrule{2-3} 
 & \multicolumn{1}{l|}{Non centered} & Centered\\ \hline
 $\mu_0, x-\mu$ & \multicolumn{1}{c|}{-} & (Geometric mean or Log-Average)\\ \hline
$\mu_1^{(2)}$ & \multicolumn{1}{c|}{$\mu$} &\multicolumn{1}{c}{$-$} \\ \hline
$\mu_2^{(3)}, x-\mu$ & \multicolumn{1}{c|}{$\mu^2+\dfrac{\sigma^2}{3}$} & \multicolumn{1}{c}{$\dfrac{\sigma^2}{3}$} \\ \hline
$\mu_3^{(4)}, x-\mu$ & \multicolumn{1}{c|}{$\mu^3+\dfrac{3\mu\sigma^2}{(4+\kappa)}$} & \multicolumn{1}{c}{$0$} \\ \hline
$\mu_4^{(5)}, x-\mu$ & \multicolumn{1}{c|}{$\mu^4+\dfrac{6\mu^2\sigma^2}{5+2\kappa}+\dfrac{3\sigma^4}{5(5+2\kappa)}$} & \multicolumn{1}{c}{$\dfrac{3\sigma^4}{25+10\kappa}$} \\ \hline
$\mu_{(n)}^{(n+1)}, x-\mu$ & \multicolumn{1}{c|}{Simplification Not Available} & \multicolumn{1}{c}{$(1+(-1^n))(\dfrac{\sigma}{\sqrt{\kappa}}^n)\dfrac{\dfrac{n-1}{2}!\dfrac{n+1-2\kappa}{2\kappa}!}{2\sqrt{\pi}\dfrac{n+1+(n-2)\kappa}{2\kappa}!}$} \\ \hline
\end{tabular}
\end{table}

\subsection{Theoretical part}\label{subsec2}
In this section, proves are provided regarding the bias and consistency of estimated parameters, namely $\sigma$ and $\kappa$ using IA.
\begin{lemma}
\begin{enumerate}
\item Suppose that the $X^{(2)}_i$, $ i=1,2,\cdots, N^{(2)}$ is independent-equal samples drawn from a 2-power coupled exponential distribution. The estimates of the scale parameter, $\hat{\sigma}$, is an unbiased estimator of $\sigma$.

\item Suppose that the $X^{(3)}_i$, $ i=1,2,\cdots, N^{(3)}$ samples from a centered 3-power coupled exponential density $f^{(3)}(x)$. The estimates of the scale $\hat{\kappa}$ is a biased estimator of $\kappa$.\label{unbias_kap}

\end{enumerate}
\end{lemma}

\begin{proof}
The proof follows the general outline of the proof of Lemma 4 \cite{nelsonIndependentApproximatesEnable2022}.
\begin{enumerate}
\item \begin{align*}
 E[2\hat{\mu}_1^{(2)}-\sigma]=&2E\left[\dfrac{1}{I^{(2)}}\sum_{i=1}^{I^{(2)}}X_i^{(2)}-\sigma\right]\\
 =&\dfrac{2}{I^{(2)}}\sum_{i=1}^{I^{(2)}}E(X_i^{(2)})-\sigma
\end{align*}
By applying $E(X_i)=\sigma/2$ from Table \ref{GPD}, the proof is completed.

\item Given the unbiased estimate \(\hat{\sigma}\), the independence of \(\hat{\mu}^{(2)}\) and \(\hat{\mu}^{(3)}\), the bias of the estimate \(\hat{\kappa}\) is given by
\begin{align*}
 \dfrac{2}{3}E\left[\dfrac{\hat{\sigma}^2}{\hat{\mu}_2^{3}}-3-\kappa\right]=&\dfrac{2\sigma^2}{3}E\left[\dfrac{1}{\dfrac{1}{I^3}\sum_{i=1}^{I^{(3)}}(X_i^{(3)}-\hat{\mu})^2}\right]-3-\kappa\\
 =&\dfrac{2\sigma^2}{3}\left[\dfrac{1}{\dfrac{2\sigma^2}{3(3+\kappa)}-\dfrac{2\sigma^2}{3I^{(2)}(3+\kappa)}}\right]-(3+\kappa)\\
 =&\dfrac{3+\kappa}{I^{(2)}-1}
\end{align*}
\end{enumerate}
\end{proof}

\begin{lemma}
\begin{enumerate}
\item 
Let $\hat{\sigma}=2\mu_1^{(2)}$ be the estimator of $\sigma$ based on the independent-equal samples $X^{(2)}_i, i=1,\dots, N^{(2)}$ with the (p.d.f) of GPD. Then $\hat{\sigma}\overset P\longrightarrow \sigma.$

\item 
 Let $\hat{\kappa}=\dfrac{2\hat{\sigma}^2}{3\hat{\mu}^{(3)}_2}-3$ be the estimator of $\kappa$ based on the independent-equal samples $X^{(3)}_i, i=1,\dots, I^{(3)}$ with the (p.d.f) of GPD. Then $\hat{\kappa}\overset P\longrightarrow \kappa.$
\end{enumerate}
\end{lemma}

\begin{proof}
\begin{enumerate}
\item 
Let $\hat{\theta}_n=\hat{\theta}_n(X_1,\dots,X_n)$ be the estimator of $\theta$ based on the random
sample $X_1,\dots,X_n$ with $p.d.f.$ $f (x, \theta)$.
The proof is based on the following theorem
\begin{theorem}\label{t1}
    An asymptotically unbiased estimator $\hat{\theta}_n$ for $\theta$ is a consistent estimator of $\theta$ if $\lim_{n\longrightarrow\infty}Var(\hat{\theta}_n)=0$ as $n\longrightarrow\infty$
\end{theorem}

According to Theorem \ref{t1}, we have
\begin{align*}
   \lim\limits_{I^{(2)}\to\infty} \operatorname{Var}\left(\hat{\sigma}_I^{(2)}\right)=& \lim\limits_{I^{(2)}\to\infty}4 \operatorname{Var}\left(\dfrac{1}{I^{(2)}}\sum_{i=1}^{I^{(2)}} X_i^{(2)}\right)\\
   =&\lim\limits_{I^{(2)}\to\infty}\dfrac{4}{I^{(2)}}\dfrac{2\sigma^2}{3(3+\kappa)}\\ 
   =&\lim\limits_{I^{(2)}\to\infty}\dfrac{8\sigma^2}{3I^{(2)}(3+\kappa)}
\end{align*}
The term in the last line tends to zero when $I^{(2)}\longrightarrow\infty$, which means that the Variances $(\hat{\sigma})$ have zero limits. The sample Variances is consistent using Theorem \ref{t1}

\item 
 From \ref{unbias_kap} [Lemma 2.1], $\hat{\kappa}$ is a biased estimator for $\kappa$. However, because $1-\dfrac{1}{I^{(2)}}\longrightarrow 1$ as $I^{(2)}\longrightarrow\infty$, we have
\begin{equation}\label{asy_unbi}
    \lim\limits_{I\to\infty} E(\hat{\kappa}_I)=\lim\limits_{I\to\infty}\left(\dfrac{3+\kappa}{1-\dfrac{1}{I^{(2)}}}\right)-3=\kappa.
\end{equation}
The $\kappa$ estimator for the $\kappa$ parameter is thus asymptotically unbiased.\\
Using the variance criterion for consistency Theorem \ref{t1} and under the assumption that the $I^{(2)} > I^{(3)}$, we have
\begin{align*}
    \lim_{I\to\infty} \text{Var}(\hat{\kappa}_I) &= \lim_{I\to\infty}\left(16~\text{Var}(\hat{\mu}^{(2)})^2 \times \left[\frac{1}{\frac{9}{I^{(3)}{}^2}\sum_{i=1}^{I^{(3)}}\left(\text{Var}(x_i^{(3)})^2+\text{Var}(\hat{\mu}^2)\right)}\right]\right)\\
\end{align*}
The variance of the square of the samples from a 3-power coupled exponential density is given by {m = 4, n = 3} in \ref{general_power}
\begin{align*}
    \dfrac{\int_{-\infty}^{\infty}x^4f^3(x)dx}{\int_{-\infty}^{\infty}f^3(x)dx}=&\left(\dfrac{\kappa}{\sigma}\right)^{-4}\dfrac{4!\left(-2-4+3+\dfrac{3}{\kappa}\right)!}{\left(-2+3+\dfrac{3}{\kappa}\right)!}\\
    =& \dfrac{8\sigma^4}{(-3+\kappa)(3+\kappa)(-3+2\kappa)}
\end{align*}

The variance of the 2-power samples is given by {m = 2, n = 2} 
\[Var(\hat{\mu})=\dfrac{\sigma^2}{I^{(2)}(2-\kappa)}, ~~~\kappa<2\]
The variance of the square of the location estimate is the square of the variance of the location estimate. Thus,

\begin{align*}
   \lim_{I\to\infty} \text{Var}(\hat{\kappa}_I) &=  \lim_{I^{(2)}\to\infty}\left[\frac{2}{\frac{9(I^{(2)}(2-\kappa))^2}{I^{(3)}}\times\left(\frac{1}{(-3+\kappa)(3+\kappa)(-3+2\kappa)}+\frac{2}{(I^{(2)}(2-\kappa))^2}\right)}\right]\\
   &=0
\end{align*}
Thus from \ref{asy_unbi} and variance criterion, the variance of the 2-power samples from a GPD is consistent with $\kappa.$
\end{enumerate}
\end{proof}

\section{Analysis of Performance and Applications}\label{subsubsec2}
\subsection{Simulation results}

In this section, we present a comprehensive evaluation of the performance of the Independent Approximates (IA) algorithm by conducting a simulation study using samples drawn from a known distribution. Our main goal is to conduct a robust empirical comparison between our method and the conventional Maximum Likelihood (ML) approach for estimation of heavy-tailed distributions, particularly the coupled exponential and coupled Gaussian. In the study, samples from a coupled exponential distribution (generalized Pareto) are drawn using (\ref{CD-PDF}) with $\alpha$ set to 1. Samples from a coupled Gaussian (Student's t) are drawn using the generalized Box-Müeller method \cite{thistletonGeneralizedBoxMuller2007, nelsonCommentsGeneralizedBoxMuller2021}.
To assess the effectiveness and accuracy of our algorithm, we employ a range of evaluation criteria, including the Coefficient of Efficiency (CE), Average Deviation (AD), Cramer-von Mises (CvM), and Negative Log-Likelihood (NLL). These criteria serve as reliable indicators to gauge the performance of our approach across different scenarios.

For the IA estimate of the scale and shape parameters for one-sided GPD distributions, we compare two approaches. In both cases the IA-pairs mean $\mu^{(2)}_1$ is one of the statistics. This is combined with either the geometric mean of original samples $\mu_0$ or the IA-triplets second-moment $\mu^{(3)}_2$. We conduct parameter estimation for various values of the shape parameter $\kappa$, spanning the range from $0.25$ to $2$, and for the sample size of $N=10,000$. Performance for samples sizes of 100 and 1000 are provided in Appendix \ref{secA2}.

To ensure the accuracy of our estimation, we develop an optimal subsampling technique aimed at minimizing estimation errors. This technique involves selecting the optimal subsample by minimizing the least absolute deviation (LAD) of the estimate, either $\mu^{(2)}_1$ or $\mu^{(3)}_2$, for selecting pairs and triplets, respectively. To find the optimal configuration, we employ a search method that varies the number of subsamples until the optimal outcome is achieved. The LAD method for determining the optimal subsample of heavy-tailed distributions offers several notable advantages. Heavy-tailed distributions often exhibit outliers or extreme observations, which can substantially impact the accuracy of the estimation process. Unlike traditional mean-based estimators that are sensitive to outliers, the LAD method is known for its robustness and resilience to extreme values. The LAD method aims to minimize the sum of the absolute deviations between the observed data points and the estimated values. The LAD method, therefore, gives more weight to outliers and extreme observations in the data, as it directly considers their absolute distances from the estimated values. Moreover, the LAD method aligns naturally with the underlying assumption of the generalized Pareto distribution (GPD), which serves as a model for heavy-tailed data. The GPD is designed to capture the tail behaviour of distributions, precisely the region where the LAD method thrives.

Tables \ref{MSE_Gau} and \ref{Goodness-t} provide a comprehensive analysis of the empirical mean square errors (MSE) and the performance metrics for different estimation methods applied to coupled Gaussian distribution. Specifically, Table \ref{MSE_Gau} details the empirical MSE for the shape $\kappa$ and  scale $\sigma$ estimations, across various values of $\kappa$.

\begin{table}[!ht]
\caption{Empirical mean square errors (MSE) of parameter estimates for data generated from a coupled Gaussian distribution and for sample size $n = 10,000$. The scale is $\sigma=0.5$ and the shape $\kappa$ varies as indicated.}\label{MSE_Gau}
\begin{tabular}{l|llll}
\hline
\multicolumn{1}{c|}{\multirow{3}{*}{$\kappa$}} & \multicolumn{4}{c}{\begin{tabular}[c]{@{}c@{}}MSE\end{tabular}}                         \\ \cmidrule{2-5} 
\multicolumn{1}{c|}{}                          & \multicolumn{2}{c|}{IA\_GM}                                         & \multicolumn{2}{c}{ML}                         \\ \cmidrule{2-5} 
\multicolumn{1}{c|}{}                          & \multicolumn{1}{l|}{$\hat{\kappa}$}       & \multicolumn{1}{l|}{$\hat{\sigma}$}       & \multicolumn{1}{l|}{$\hat{\kappa}$}       & $\hat{\sigma}$       \\ \hline
0.25                                           & \multicolumn{1}{l|}{$0.080\pm 0.005$} & \multicolumn{1}{l|}{$0.013\pm 0.003$} & \multicolumn{1}{l|}{$0.008\pm 0.007$} & $0.006\pm 0.005$ \\ \hline
0.5                                            & \multicolumn{1}{l|}{$0.006\pm 0.003$} & \multicolumn{1}{l|}{$0.012\pm 0.003$} & \multicolumn{1}{l|}{$0.007\pm0.004$} & $0.007\pm0.006$ \\ \hline
1                                              & \multicolumn{1}{l|}{$0.034\pm0.009$} & \multicolumn{1}{l|}{$0.005\pm0.009$} & \multicolumn{1}{l|}{$0.01\pm0.01$} & $0.009\pm0.003$ \\ \hline
1.25                                           & \multicolumn{1}{l|}{$0.03\pm0.01$} & \multicolumn{1}{l|}{$0.001\pm0.003$} & \multicolumn{1}{l|}{$0.008\pm0.009$} & $0.009\pm0.004$ \\ \hline
2                                              & \multicolumn{1}{l|}{$0.06\pm0.02$} & \multicolumn{1}{l|}{$0.040\pm0.003$} & \multicolumn{1}{l|}{$0.01\pm0.02$} & $0.013\pm0.007$ \\ \hline
\end{tabular}
\end{table}

\begin{table}[h]
\caption{Goodness-of-Fit Metrics for the Coupled Gaussian Distribution under Various Methods and Shape Parameters $\kappa$ with a Fixed Scale Parameter ($\sigma = 0.5$).}
\label{Goodness-t}
\begin{tabular}{llllll}
\hline
\multicolumn{6}{c}{Coupled Gaussian} \\ \hline
\multicolumn{6}{c}{Average deviation (AD)    {$\sigma=0.5$}} \\ \hline
\multicolumn{1}{l|}{Method\textbackslash $\kappa$} & \multicolumn{1}{l|}{0.25} & \multicolumn{1}{l|}{0.5} & \multicolumn{1}{l|}{1} & \multicolumn{1}{l|}{1.25} & 2 \\ \hline
\multicolumn{1}{l|}{IA (Geometric mean)} & \multicolumn{1}{l|}{0.03} & \multicolumn{1}{l|}{0.063} & \multicolumn{1}{l|}{0.71} & \multicolumn{1}{l|}{4.3} & 2200 \\ \hline
\multicolumn{1}{l|}{ML} & \multicolumn{1}{l|}{0.034} & \multicolumn{1}{l|}{0.097} & \multicolumn{1}{l|}{17} & \multicolumn{1}{l|}{300} & 620,000 \\ \hline
\multicolumn{6}{c}{Cramer–von Mises (CvM)  $\sigma=0.5$} \\ \hline
\multicolumn{1}{l|}{Method\textbackslash $\kappa$} & \multicolumn{1}{l|}{0.25} & \multicolumn{1}{l|}{0.5} & \multicolumn{1}{l|}{1} & \multicolumn{1}{l|}{1.25} & 2 \\ \hline
\multicolumn{1}{l|}{IA (Geometric mean)} & \multicolumn{1}{l|}{0.76} & \multicolumn{1}{l|}{0.75} & \multicolumn{1}{l|}{0.74} & \multicolumn{1}{l|}{0.74} & 0.79 \\ \hline
\multicolumn{1}{l|}{ML} & \multicolumn{1}{l|}{1.1} & \multicolumn{1}{l|}{1.1} & \multicolumn{1}{l|}{1.0} & \multicolumn{1}{l|}{1.0} & 0.96 \\ \hline
\multicolumn{6}{c}{NLL} \\ \hline
\multicolumn{1}{l|}{Method\textbackslash $\kappa$} & \multicolumn{1}{l|}{0.25} & \multicolumn{1}{l|}{0.5} & \multicolumn{1}{l|}{1} & \multicolumn{1}{l|}{1.25} & 2 \\ \hline
\multicolumn{1}{l|}{IA (Geometric mean)} & \multicolumn{1}{l|}{16,000} & \multicolumn{1}{l|}{16,000} & \multicolumn{1}{l|}{19,000} & \multicolumn{1}{l|}{22,000} & 38,000 \\ \hline
\multicolumn{1}{l|}{ML} & \multicolumn{1}{l|}{17,000} & \multicolumn{1}{l|}{15,000} & \multicolumn{1}{l|}{19,000} & \multicolumn{1}{l|}{22,000} & 39,000 \\ \hline
\end{tabular}
\end{table}

Table \ref{Goodness-t} presents insights into the performance of estimation methods for the coupled Gaussian distribution using Average Deviation (AD), Cramer–von Mises (CvM), and and Negative Log-Likelihood (NLL) as criteria.
The IA\_GM method exhibits a distinct performance when compared to the quality metrics used. It appears to be more stable and capable of achieving a better fit with the data in general.  This is particularly noticeable in the AD method, where the IA\_GM estimator performs significantly better than the ML estimator, especially when $\kappa$ is high.

In the context of the coupled exponential distribution, Table \ref{MSE_GPD} provides the empirical MSE  for the shape $\kappa$ and  scale $\sigma$ parameters. Our analysis encompassed the IA\_GM, IA, and ML methods. The results indicate that the IA\_GM and IA methods exhibit superior performance compared to the ML method.

\begin{table}[!ht]
\caption{Empirical mean square errors (MSE) of parameter estimates for data generated from a coupled exponential distribution with a sample size $n = 10,000$. The scale $\sigma=0.5$ and the shape $\kappa$ varies as indicated.}\label{MSE_GPD}
\begin{tabular}{l|llllll}
\hline
\multicolumn{1}{c|}{\multirow{3}{*}{$\kappa$}} & \multicolumn{6}{c}{\begin{tabular}[c]{@{}c@{}}$(MSE\pm SD)$x$10^{-3}$\end{tabular}}                                  \\ \cmidrule{2-7} 
\multicolumn{1}{c|}{}                          & \multicolumn{2}{c|}{IA\_GM}                             & \multicolumn{2}{c|}{IA}                                 & \multicolumn{2}{l}{ML}             \\ \cmidrule{2-7} 
\multicolumn{1}{c|}{}                          & \multicolumn{1}{l|}{$\hat{\kappa}$} & \multicolumn{1}{l|}{$\hat{\sigma}$} & \multicolumn{1}{l|}{$\hat{\kappa}$} & \multicolumn{1}{l|}{$\hat{\sigma}$} & \multicolumn{1}{l|}{$\hat{\kappa}$} & $\hat{\sigma}$ \\ \hline
0.25& \multicolumn{1}{l|}{$0\pm1$} & \multicolumn{1}{l|}{$1.0\pm0.1$} & \multicolumn{1}{l|}{$6\pm6$} & \multicolumn{1}{l|}{$9\pm4$} & \multicolumn{1}{l|}{$22\pm5$} & $6\pm5$ \\ \hline
0.5& \multicolumn{1}{l|}{$0\pm5$} & \multicolumn{1}{l|}{$1.0\pm0.2$} & \multicolumn{1}{l|}{$0\pm4$} & \multicolumn{1}{l|}{$0\pm3$} & \multicolumn{1}{l|}{$20\pm7$} & $4\pm5$ \\ \hline
1& \multicolumn{1}{l|}{$0\pm10$}  & \multicolumn{1}{l|}{$1.0\pm0.2$}& \multicolumn{1}{l|}{$20\pm10$} & \multicolumn{1}{l|}{$32\pm3$} & \multicolumn{1}{l|}{$20\pm10$} & $4\pm5$ \\ \hline
1.25& \multicolumn{1}{l|}{$3\pm10$} & \multicolumn{1}{l|}{$1.0\pm0.2$} & \multicolumn{1}{l|}{$3\pm9$} & \multicolumn{1}{l|}{$15\pm4$} & \multicolumn{1}{l|}{$20\pm10$} & $3\pm3$ \\ \hline
2& \multicolumn{1}{l|}{$20\pm20$} & \multicolumn{1}{l|}{$5.0\pm0.2$} & \multicolumn{1}{l|}{$60\pm10$} & \multicolumn{1}{l|}{$39\pm6$} & \multicolumn{1}{l|}{$10\pm20$} & $13\pm2$ \\ \hline
\end{tabular}
\end{table}

\begin{table}[!ht]
\caption{Goodness-of-Fit Metrics for the Coupled Exponential Distribution under Various Methods and Shape Parameters $\kappa$ with a Fixed Scale Parameter ($\sigma = 0.5$)}
\label{Goodness-GP}
\begin{tabular}{llllll}
\hline
\multicolumn{6}{c}{Coupled Exponential} \\ \hline
\multicolumn{6}{c}{Average deviation (AD)    sigma=0.5} \\ \hline
\multicolumn{1}{l|}{Method\textbackslash kappa} & \multicolumn{1}{l|}{0.25} & \multicolumn{1}{l|}{0.5} & \multicolumn{1}{l|}{1} & \multicolumn{1}{l|}{1.25} & 2 \\ \hline
\multicolumn{1}{l|}{IA (Geometric mean)} & \multicolumn{1}{l|}{$0.14$x$10^{-3}$} & \multicolumn{1}{l|}{$0.53$x$10^{-3}$} & \multicolumn{1}{l|}{0.003} & \multicolumn{1}{l|}{0.17} & 320 \\ \hline
\multicolumn{1}{l|}{IA (Triplets)} & \multicolumn{1}{l|}{0.007} & \multicolumn{1}{l|}{0.002} & \multicolumn{1}{l|}{0.15} & \multicolumn{1}{l|}{0.14} & 650 \\ \hline
\multicolumn{1}{l|}{ML} & \multicolumn{1}{l|}{0.013} & \multicolumn{1}{l|}{0.027} & \multicolumn{1}{l|}{0.33} & \multicolumn{1}{l|}{1.57} & 180 \\ \hline
\multicolumn{6}{c}{Cramer–von Mises (CvM)  Sigma=0.5} \\ \hline
\multicolumn{1}{l|}{Method\textbackslash kappa} & \multicolumn{1}{l|}{0.25} & \multicolumn{1}{l|}{0.5} & \multicolumn{1}{l|}{1} & \multicolumn{1}{l|}{1.25} & 2 \\ \hline
\multicolumn{1}{l|}{IA (Geometric mean)} & \multicolumn{1}{l|}{$0.045$x$10^{-3}$} & \multicolumn{1}{l|}{$0.086$x$10^{-3}$} & \multicolumn{1}{l|}{0.0002} & \multicolumn{1}{l|}{0.001} & 0.007 \\ \hline
\multicolumn{1}{l|}{IA (Triplets)} & \multicolumn{1}{l|}{0.063} & \multicolumn{1}{l|}{0.0005} & \multicolumn{1}{l|}{0.55} & \multicolumn{1}{l|}{0.11} & 0.27 \\ \hline
\multicolumn{1}{l|}{ML} & \multicolumn{1}{l|}{0.013} & \multicolumn{1}{l|}{0.009} & \multicolumn{1}{l|}{0.01} & \multicolumn{1}{l|}{0.011} & 0.078 \\ \hline
\multicolumn{6}{c}{NLL} \\ \hline
\multicolumn{1}{l|}{Method\textbackslash kappa} & \multicolumn{1}{l|}{0.25} & \multicolumn{1}{l|}{0.5} & \multicolumn{1}{l|}{1} & \multicolumn{1}{l|}{1.25} & 2 \\ \hline
\multicolumn{1}{l|}{IA (Geometric mean)} & \multicolumn{1}{l|}{5,500} & \multicolumn{1}{l|}{8,000} & \multicolumn{1}{l|}{13,000} & \multicolumn{1}{l|}{15,000} & 23,000 \\ \hline
\multicolumn{1}{l|}{IA (Triplets)} & \multicolumn{1}{l|}{5,500} & \multicolumn{1}{l|}{8,000} & \multicolumn{1}{l|}{13,000} & \multicolumn{1}{l|}{16,000} & 23,000 \\ \hline
\multicolumn{1}{l|}{ML} & \multicolumn{1}{l|}{5,600} & \multicolumn{1}{l|}{8,100} & \multicolumn{1}{l|}{13,000} & \multicolumn{1}{l|}{16,000} & 23,000 \\ \hline
\end{tabular}
\end{table}

The performance of various estimation methods for the coupled exponential distribution is elucidated in Table \ref{Goodness-GP} using AD, CvM, and and NLL as criteria. Notably, in terms of methodological efficacy, exemplary results are observed. For instance, at $\kappa = 0.25$, the IA\_GM and IA methods exhibit notably smaller AD values, indicating their superior goodness of fit. This trend persists as $\kappa$ increases, with IA\_GM and IA consistently outperforming ML in the Cramer–von Mises (CvM) tests. These results highlight the robustness of IA\_GM and IA in capturing the nuances of the coupled exponential distribution. In contrast, the ML method tends to underperform across the range of $\kappa$ values, underscoring its limitations in accurately capturing the characteristics of the distribution.

\subsection{Analysis of Coherent Noise Model (CNM) Results}\label{sec4}
To evaluate the effectiveness of the IA estimator with samples from a coupled exponential distribution with unknown parameters, we've utilized data from a simulation of CNM \cite{celikogluAnalysisReturnDistributions2010}. 
This model is designed to provide a structured framework for the evaluation of a system's response to external stressors \cite{newmanAvalanchesScalingCoherent1996, sneppenCoherentNoiseScale1997}. The model involves a collection of $N$ agents, each distinguished by a unique threshold denoted as $x_i$. These thresholds serve as indicators of the agents' capacity to withstand external stress, represented by the variable $\eta$. Importantly, both the threshold values and external stress factors are drawn from probability distributions, namely $p_{threshold}(x)$ and $p_{stress}(\eta)$.
The stress is modeled as an exponential distribution $p_{stress}(\eta) = (1/\eta) \exp(-\eta/\sigma)$. The threshold is modeled as a uniform distribution between 0 to 1, $p_{threshold}(x)$:
\[
p_{threshold}(x) = 
\begin{cases} 
1 & \text{if } 0 \le x \le 1, \\
0 & \text{otherwise}
\end{cases}
.\]
We used $\sigma = 0.05$ (threshold of exponential cut-off) , $f=8000$ and a time series of $4\times10^8$ sampled in $N=10,000$ increments.

The dynamics of the model is very simple, yet effective (for the details of the model, see \cite{wilkeAftershocksCoherentnoiseModels1998, sarlisPredictabilityCoherentnoiseModel2012}, :
\begin{itemize}
    \item A random stress factor, $\eta$, is generated in accordance with the $p_{stress}(\eta)$ distribution. Agents with threshold values $(x_i)$ falling below $\eta$ are systematically replaced by new agents, each of whom is endowed with a threshold drawn from the $p_{threshold}(x)$ distribution.

    \item To maintain the continuity of avalanche generation, a fraction of the $N$ agents, denoted as $f$, undergo threshold updates. New threshold values for these selected agents are drawn once again from the $p_{threshold}(x)$ distribution.
    
    \item This process iterates, with the first step being executed in each subsequent time interval.
\end{itemize}

Tables \ref{est_CNM} and Figure \ref{analysis_CNM} present the parameter estimates and the performance evaluation of estimation methods, utilizing AD, CvM, and NLL as criteria for a Coherent Noise Model (CNM) based on a coupled exponential distribution. In all three criteria, we note that the IA\_GM estimator has better performance than ML and IA. Therefore, based on the provided results, IA\_GM appears to be the best method for estimating the coupled exponential distribution for the CNM.

\begin{table}[h]
    \centering
   \caption{Comparison of CNM Parameter Estimates using three methods: ML, IA\_GM, and IA}
\label{est_CNM}
    \begin{tabular}{lllll}
      \hline
      \multicolumn{5}{c}{\begin{tabular}[c]{@{}c@{}}Analyzing Standard Map set Results using \\ a coupled exponential\end{tabular}} \\ \hline
      \multicolumn{1}{l|}{Method} & \multicolumn{1}{l|}{$\hat{\kappa}$} & \multicolumn{1}{l|}{$\hat{\sigma}(10^{-3})$} & \multicolumn{1}{l|}{$\hat{q}$} & {$\hat{\beta}$}  \\ \hline
      \multicolumn{1}{l|}{$ML$} & \multicolumn{1}{l|}{$0.961\pm 0.004$
      } & \multicolumn{1}{l|}{$4.60\pm 0.001$
      } & \multicolumn{1}{l|}{1.49} & {4300} \\ \hline
      \multicolumn{1}{l|}{$IA\_GM$} & \multicolumn{1}{l|}{$0.92\pm 0.05$
      } & \multicolumn{1}{l|}{$4.9\pm0.02$
      } & \multicolumn{1}{l|}{1.5} & {3900} \\ \hline
      \multicolumn{1}{l|}{$IA$} & \multicolumn{1}{l|}{$0.97\pm0.01$
      } & \multicolumn{1}{c|}{$4.90\pm 0.06$
} & \multicolumn{1}{c|}{1.5} & {4000} \\ \hline
    \end{tabular}
\end{table}

\begin{figure}[!ht]
 \caption{(a) Evaluation Criteria for Goodness-of-Fit across Three Methods (ML, IA\_GM, and IA) in the Coupled Exponential Distribution: The evaluation employs various approaches: AD, CvM, and NLL for the CNM model.
 (b) Histogram of the Estimated PDF for the Coupled Exponential Distribution: The figure displays the original CNM model along with the estimations using the IA, GM, and ML methods.}
    \label{fig:combined}
    \centering
    \begin{subfigure}[b]{0.7\textwidth} 
        \centering
        \scalebox{1}{
        \begin{tabular}{lccc}
            \hline
            \multicolumn{4}{c}{\begin{tabular}[c]{@{}c@{}}Analyzing Coherent Noise Model (CNM) Results\\ using a coupled exponential \end{tabular}} \\ \hline
            \multicolumn{1}{l|}{Methods/ Criteria} & \multicolumn{1}{c|}{AD} & \multicolumn{1}{c|}{CvM} & NLL \\ \hline
            \multicolumn{1}{l|}{ML} & \multicolumn{1}{c|}{0.005} & \multicolumn{1}{c|}{0.70} & 60,000 \\ \hline
            \multicolumn{1}{l|}{$IA\_GM$} & \multicolumn{1}{c|}{0.001} & \multicolumn{1}{c|}{0.68} & 60,000 \\ \hline
            \multicolumn{1}{l|}{IA} & \multicolumn{1}{c|}{0.001} & \multicolumn{1}{c|}{1.1} & 60,000 \\ \hline
        \end{tabular}}
        \caption{}
        \label{analysis_CNM}
    \end{subfigure}
\end{figure}

\begin{figure}[!ht]\ContinuedFloat
    \centering
    \begin{subfigure}[b]{0.5\textwidth}
        \centering
        \includegraphics[width=8cm,height=6cm]{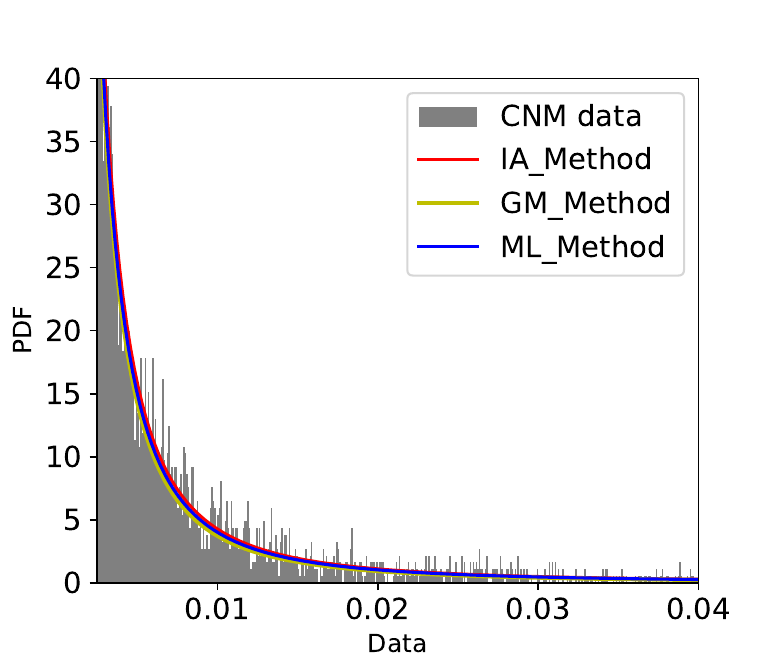}
        \caption{}
        \label{fig:CNM}
    \end{subfigure}
\end{figure}

\subsection{The standard map model}\label{sec6}

To assess the performance of the IA estimator with samples from a coupled Gaussian distribution with unknown parameters we utilized a simulation of the standard map. It is a well-known two-dimensional conservative nonlinear dynamical system described by an iterative function of two variables 
\begin{equation}
\begin{array}{l}
y_{i+1} =p_{i} - K \sin(x_i)\\
x_{i+1}=x_{i} + y_{i+1}
\end{array}
\label{gen-stan-map}
\end{equation}
where $x$ and $y$ are taken as modulo $2\pi$ \cite{zaslavskyHamiltonianChaosFractional2005, izraelevNearlyLinearMappings1980, chirikovUniversalInstabilityManydimensional1979}. 
It has already been numerically shown \cite{tirnakliStandardMapBoltzmannGibbs2016} that, for small $K$ values (for which the phase space is dominated by the stability islands), central limit behavior of the model can be well approximated by a $q$-Gaussian with $q\simeq 1.935$, defining the shape based on its relationship with the Tsallis parameter outlined in Section \ref{background}. For more details refer to \cite{nelsonIndependentApproximatesEnable2022}.

Figure \ref{analysis_map}  provides a comprehensive comparison of parameter estimates for $\kappa$ and $\sigma$ using the IA\_GM and ML methods, along with an evaluation of their performance based on the criteria AD, CvM, and NLL.  
The AD performance of IA\_GM excels indicating a superior fit. CvM values further support this trend, with IA\_GM presenting a lower CvM value, suggesting a better overall fit. Additionally, NLL values corroborate the superiority of IA\_GM, as it achieves a lower NLL value. Based on these criteria, the IA\_GM method stands out as a more effective choice compared to the ML method for the specified Standard Map Set with a coupled Gaussian distribution.

\begin{figure}[!ht]
\caption{(a) Comparison of Standard Map Parameter Estimates and Evaluation Criteria for Goodness-of-Fit across two Methods (ML, IA\_GM) in the Coupled Gaussian Distribution Using Various Approaches (AD, CvM, and NLL for the Map Model). (b) Histogram of the Estimated PDF for the Coupled Gaussian Distribution. The Figure displays the original map data along with the estimations using the Independent Approximates algorithm (IA) and the Maximum Likelihood (ML) method.}
\label{Map_set}
\centering
\begin{subfigure}[b]{1\textwidth} 
    \centering
    \scalebox{1}{
    \begin{tabular}[b]{llllllll}
        \hline
        \multicolumn{8}{c}{\begin{tabular}[c]{@{}c@{}}Analyzing Standard Map set Results using \\ a coupled Gaussian\end{tabular}} \\ \hline
        \multicolumn{1}{l|}{Method} & \multicolumn{1}{l|}{$\hat{\kappa}$} & \multicolumn{1}{l|}{$\hat{\sigma}$} & \multicolumn{1}{l|}{$\hat{q}$} & \multicolumn{1}{l|}{$\hat{\beta}$} & \multicolumn{1}{l|}{AD} & \multicolumn{1}{l|}{CvM} & NLL \\ \hline
        \multicolumn{1}{l|}{$IA\_GM$} & \multicolumn{1}{l|}{$0.91\pm0.05$} & \multicolumn{1}{l|}{$0.076\pm0.001$} &\multicolumn{1}{l|}{1.953} & \multicolumn{1}{l|}{159} & \multicolumn{1}{l|}{0.43} & \multicolumn{1}{l|}{11} & 9900 \\ \hline
        \multicolumn{1}{l|}{$ML$} & \multicolumn{1}{l|}{$0.900\pm0.007$} & \multicolumn{1}{l|}{$0.080\pm0.001$} & \multicolumn{1}{l|}{1.947} & \multicolumn{1}{l|}{146} &\multicolumn{1}{l|}{1.6} & \multicolumn{1}{l|}{12} & 10,000 \\ \hline
    \end{tabular}}
    \caption{}
    \label{analysis_map}
\end{subfigure}
\end{figure}

\begin{figure}[!ht]\ContinuedFloat
\centering
\begin{subfigure}[b]{0.6\textwidth}
    \centering
    \includegraphics[width=8cm,height=5.5cm]{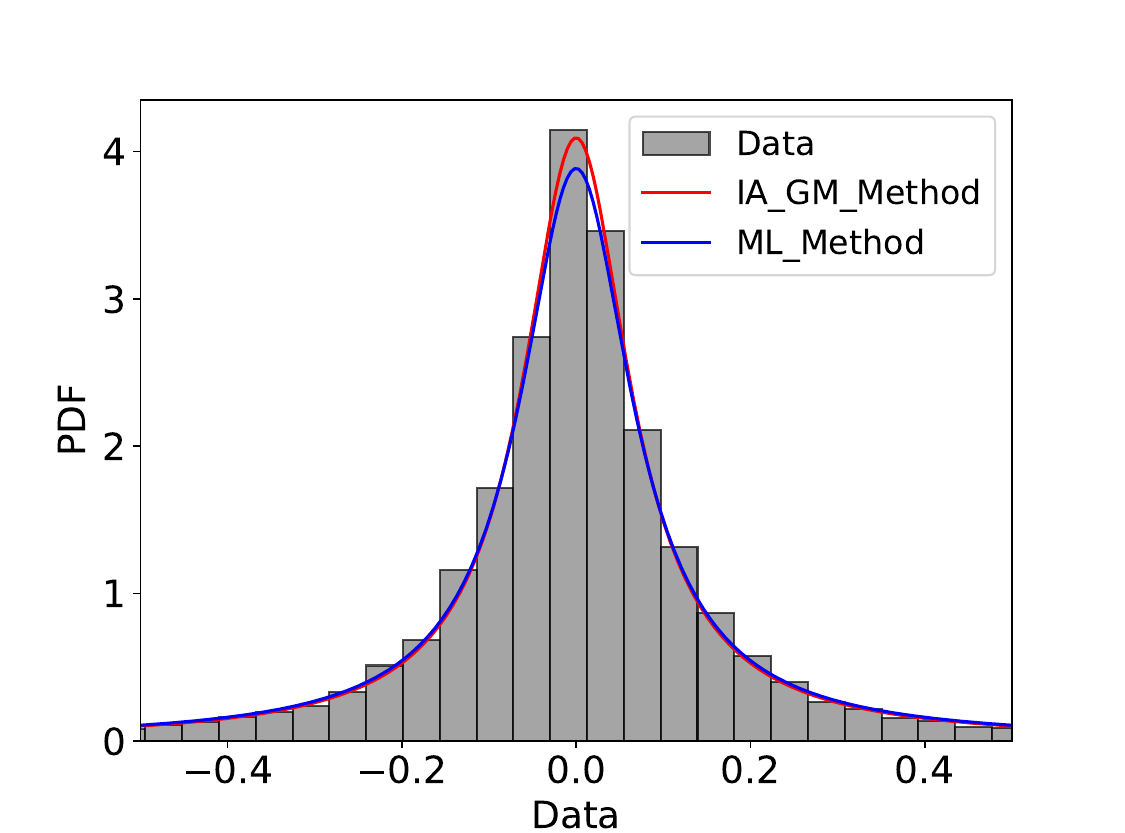}
    \caption{}
    \label{fig:map}
\end{subfigure}
\end{figure}

\FloatBarrier

\section{Conclusion}
In this paper, we improve and evaluate the performance of the Independent Approximates (IA), a novel estimation method for obtaining accurate statistical measures from heavy-tailed distributions, particularly within the framework of the coupled exponential family of distributions. We proved that IA may be used as a maximum likelihood estimator for the generalized Pareto and Student's $t$ distributions by grouping independent samples into n-tuples and selecting those that are approximately equal. The basis of the IA method is the estimation of the moments of the 2nd and 3rd powers of the underlying distribution, which guaranteed to have first and second-order moments, respectively. Samples from the 2nd (or 3rd) power are obtained by filtering pairs (or triplets) that are approximately equal from the original distribution.

The IA methodology provides evidence that the escort distribution of NSM is the distribution of \textit{q} random variables sharing the same state. Although NSM has made progress toward a statistical theory of complex systems, its reliance on a secondary property, the \textit{q} number of independent equals, has limited its interpretative ability. Here we define the heavy-tailed distributions and the statistical methods in terms of the degree of nonlinear coupling, \(\kappa\), which is also the distribution shape parameter. Given the direct connection between the nonlinearity of a system and its complexity, the coupling could be considered a measure of the statistical complexity of a system. Clarity about the physical model represented by the Independent-Equals distribution is expected to expand its applications, including recent improvements in the training of robust machine learning algorithms.

Simulation studies and empirical applications, such as the Coherent Noise Model (CNM) and the Standard Map, showed that the IA technique was robustness and accuracy comparable to traditional maximum likelihood approaches. In particular, the IA approach demonstrated competitive performance, particularly in scenarios in which strong nonlinearity leads to limit distributions with very slow decaying tails. Future research could explore the extension of IA methods to multidimensional distributions, incorporate dynamic thresholds in time-evolving systems, and investigate its use in real-world datasets from domains such as finance, neuroscience, and geophysics, where heavy tails and complexity are prevalent.
\section*{Acknowledgements} 
U.T. is a member of the Science Academy, Bilim Akademisi, Turkey and acknowledges partial support from 
TUBITAK (Turkish Agency) under the Research Project number 121F269. We thank Christian Beck and Grzegorz Wilk for correspondence on the derivation of superstatistics.

\bibliography{MICS}


\begin{thebibliography}{33}
\ifx \bisbn   \undefined \def \bisbn  #1{ISBN #1}\fi
\ifx \binits  \undefined \def \binits#1{#1}\fi
\ifx \bauthor  \undefined \def \bauthor#1{#1}\fi
\ifx \batitle  \undefined \def \batitle#1{#1}\fi
\ifx \bjtitle  \undefined \def \bjtitle#1{#1}\fi
\ifx \bvolume  \undefined \def \bvolume#1{\textbf{#1}}\fi
\ifx \byear  \undefined \def \byear#1{#1}\fi
\ifx \bissue  \undefined \def \bissue#1{#1}\fi
\ifx \bfpage  \undefined \def \bfpage#1{#1}\fi
\ifx \blpage  \undefined \def \blpage #1{#1}\fi
\ifx \burl  \undefined \def \burl#1{\textsf{#1}}\fi
\ifx \doiurl  \undefined \def \doiurl#1{\url{https://doi.org/#1}}\fi
\ifx \betal  \undefined \def \betal{\textit{et al.}}\fi
\ifx \binstitute  \undefined \def \binstitute#1{#1}\fi
\ifx \binstitutionaled  \undefined \def \binstitutionaled#1{#1}\fi
\ifx \bctitle  \undefined \def \bctitle#1{#1}\fi
\ifx \beditor  \undefined \def \beditor#1{#1}\fi
\ifx \bpublisher  \undefined \def \bpublisher#1{#1}\fi
\ifx \bbtitle  \undefined \def \bbtitle#1{#1}\fi
\ifx \bedition  \undefined \def \bedition#1{#1}\fi
\ifx \bseriesno  \undefined \def \bseriesno#1{#1}\fi
\ifx \blocation  \undefined \def \blocation#1{#1}\fi
\ifx \bsertitle  \undefined \def \bsertitle#1{#1}\fi
\ifx \bsnm \undefined \def \bsnm#1{#1}\fi
\ifx \bsuffix \undefined \def \bsuffix#1{#1}\fi
\ifx \bparticle \undefined \def \bparticle#1{#1}\fi
\ifx \barticle \undefined \def \barticle#1{#1}\fi
\bibcommenthead
\ifx \bconfdate \undefined \def \bconfdate #1{#1}\fi
\ifx \botherref \undefined \def \botherref #1{#1}\fi
\ifx \url \undefined \def \url#1{\textsf{#1}}\fi
\ifx \bchapter \undefined \def \bchapter#1{#1}\fi
\ifx \bbook \undefined \def \bbook#1{#1}\fi
\ifx \bcomment \undefined \def \bcomment#1{#1}\fi
\ifx \oauthor \undefined \def \oauthor#1{#1}\fi
\ifx \citeauthoryear \undefined \def \citeauthoryear#1{#1}\fi
\ifx \endbibitem  \undefined \def \endbibitem {}\fi
\ifx \bconflocation  \undefined \def \bconflocation#1{#1}\fi
\ifx \arxivurl  \undefined \def \arxivurl#1{\textsf{#1}}\fi
\csname PreBibitemsHook\endcsname

\bibitem[\protect\citeauthoryear{Abe
  et~al.}{2001}]{abeNonextensiveStatisticalMechanics2001a}
\begin{bbook}
\beditor{\bsnm{Abe}, \binits{S.}},
\beditor{\bsnm{Okamoto}, \binits{Y.}},
\beditor{\bsnm{Beig}, \binits{R.}},
\beditor{\bsnm{Ehlers}, \binits{J.}},
\beditor{\bsnm{Frisch}, \binits{U.}},
\beditor{\bsnm{Hepp}, \binits{K.}},
\beditor{\bsnm{Hillebrandt}, \binits{W.}},
\beditor{\bsnm{Imboden}, \binits{D.}},
\beditor{\bsnm{Jaffe}, \binits{R.L.}},
\beditor{\bsnm{Kippenhahn}, \binits{R.}},
\beditor{\bsnm{Lipowsky}, \binits{R.}},
\beditor{\bsnm{L{\"o}hneysen}, \binits{H.V.}},
\beditor{\bsnm{Ojima}, \binits{I.}},
\beditor{\bsnm{Weidenm{\"u}ller}, \binits{H.A.}},
\beditor{\bsnm{Wess}, \binits{J.}},
\beditor{\bsnm{Zittartz}, \binits{J.}} (eds.):
\bbtitle{Nonextensive {{Statistical Mechanics}} and {{Its Applications}}}.
\bsertitle{Lecture {{Notes}} in {{Physics}}},
vol. \bseriesno{560}.
\bpublisher{Springer},
\blocation{Berlin, Heidelberg}
(\byear{2001}).
\doiurl{10.1007/3-540-40919-X}
\end{bbook}
\endbibitem

\bibitem[\protect\citeauthoryear{Anteneodo and
  Tsallis}{2003}]{anteneodoMultiplicativeNoiseNonextensive2003}
\begin{barticle}
\bauthor{\bsnm{Anteneodo}, \binits{C.}},
\bauthor{\bsnm{Tsallis}, \binits{C.}}:
\batitle{Multiplicative noise: A mechanism leading to nonextensive statistical
  mechanics}.
\bjtitle{Journal of Mathematical Physics}
\bvolume{44}(\bissue{11}),
\bfpage{5194}--\blpage{5203}
(\byear{2003})
\doiurl{10.1063/1.1613633}
\end{barticle}
\endbibitem

\bibitem[\protect\citeauthoryear{Beck and
  Cohen}{2003}]{beckSuperstatistics2003}
\begin{barticle}
\bauthor{\bsnm{Beck}, \binits{C.}},
\bauthor{\bsnm{Cohen}, \binits{E.G.D.}}:
\batitle{Superstatistics}.
\bjtitle{Physica A: Statistical Mechanics and its Applications}
\bvolume{322},
\bfpage{267}--\blpage{275}
(\byear{2003})
\doiurl{10.1016/S0378-4371(03)00019-0}
\end{barticle}
\endbibitem

\bibitem[\protect\citeauthoryear{Beck}{2001}]{beckDynamicalFoundationsNonextensive2001}
\begin{barticle}
\bauthor{\bsnm{Beck}, \binits{C.}}:
\batitle{Dynamical {{Foundations}} of {{Nonextensive Statistical Mechanics}}}.
\bjtitle{Physical Review Letters}
\bvolume{87}(\bissue{18}),
\bfpage{180601}
(\byear{2001})
\doiurl{10.1103/PhysRevLett.87.180601}
\end{barticle}
\endbibitem

\bibitem[\protect\citeauthoryear{Bradley and
  Taqqu}{2003}]{bradleyFinancialRiskHeavy2003}
\begin{bchapter}
\bauthor{\bsnm{Bradley}, \binits{B.O.}},
\bauthor{\bsnm{Taqqu}, \binits{M.S.}}:
\bctitle{Financial risk and heavy tails}.
In: \bbtitle{Handbook of Heavy Tailed Distributions in Finance},
pp. \bfpage{35}--\blpage{103}.
\bpublisher{Elsevier}, \blocation{???}
(\byear{2003})
\end{bchapter}
\endbibitem

\bibitem[\protect\citeauthoryear{Chirikov}{1979}]{chirikovUniversalInstabilityManydimensional1979}
\begin{barticle}
\bauthor{\bsnm{Chirikov}, \binits{B.V.}}:
\batitle{A universal instability of many-dimensional oscillator systems}.
\bjtitle{Physics reports}
\bvolume{52}(\bissue{5}),
\bfpage{263}--\blpage{379}
(\byear{1979})
\end{barticle}
\endbibitem

\bibitem[\protect\citeauthoryear{Celikoglu
  et~al.}{2010}]{celikogluAnalysisReturnDistributions2010}
\begin{barticle}
\bauthor{\bsnm{Celikoglu}, \binits{A.}},
\bauthor{\bsnm{Tirnakli}, \binits{U.}},
\bauthor{\bsnm{Queir{\'o}s}, \binits{S.M.D.}}:
\batitle{Analysis of return distributions in the coherent noise model}.
\bjtitle{Physical Review E}
\bvolume{82}(\bissue{2}),
\bfpage{021124}
(\byear{2010})
\doiurl{10.1103/PhysRevE.82.021124}
\end{barticle}
\endbibitem

\bibitem[\protect\citeauthoryear{Dekkers et~al.}{1989}]{dekkers1989moment}
\begin{barticle}
\bauthor{\bsnm{Dekkers}, \binits{A.L.M.}},
\bauthor{\bsnm{Einmahl}, \binits{J.H.J.}},
\bauthor{\bsnm{Haan}, \binits{L.}}:
\batitle{A moment estimator for the index of an extreme-value distribution}.
\bjtitle{The Annals of Statistics}
\bvolume{17}(\bissue{4}),
\bfpage{1833}--\blpage{1855}
(\byear{1989})
\doiurl{10.1214/aos/1176347397}
\end{barticle}
\endbibitem

\bibitem[\protect\citeauthoryear{Granha
  et~al.}{2022}]{granhaOpinionDynamicsFinancial2022}
\begin{barticle}
\bauthor{\bsnm{Granha}, \binits{M.F.}},
\bauthor{\bsnm{Vilela}, \binits{A.L.}},
\bauthor{\bsnm{Wang}, \binits{C.}},
\bauthor{\bsnm{Nelson}, \binits{K.P.}},
\bauthor{\bsnm{Stanley}, \binits{H.E.}}:
\batitle{Opinion dynamics in financial markets via random networks}.
\bjtitle{Proceedings of the National Academy of Sciences}
\bvolume{119}(\bissue{49}),
\bfpage{2201573119}
(\byear{2022})
\end{barticle}
\endbibitem

\bibitem[\protect\citeauthoryear{Hill}{1975}]{hill1975tail}
\begin{barticle}
\bauthor{\bsnm{Hill}, \binits{B.M.}}:
\batitle{A simple general approach to inference about the tail of a
  distribution}.
\bjtitle{The Annals of Statistics}
\bvolume{3}(\bissue{5}),
\bfpage{1163}--\blpage{1174}
(\byear{1975})
\doiurl{10.1214/aos/1176343247}
\end{barticle}
\endbibitem

\bibitem[\protect\citeauthoryear{Ibragimov
  et~al.}{2015}]{ibragimovHeavyTailedDistributionsRobustness2015}
\begin{bbook}
\bauthor{\bsnm{Ibragimov}, \binits{M.}},
\bauthor{\bsnm{Ibragimov}, \binits{R.}},
\bauthor{\bsnm{Walden}, \binits{J.}}:
\bbtitle{Heavy-{{Tailed Distributions}} and {{Robustness}} in {{Economics}} and
  {{Finance}}}.
\bsertitle{Lecture {{Notes}} in {{Statistics}}},
vol. \bseriesno{214}.
\bpublisher{Springer},
\blocation{Cham}
(\byear{2015}).
\doiurl{10.1007/978-3-319-16877-7}
\end{bbook}
\endbibitem

\bibitem[\protect\citeauthoryear{Izraelev}{1980}]{izraelevNearlyLinearMappings1980}
\begin{barticle}
\bauthor{\bsnm{Izraelev}, \binits{F.M.}}:
\batitle{Nearly linear mappings and their applications}.
\bjtitle{Physica D: Nonlinear Phenomena}
\bvolume{1}(\bissue{3}),
\bfpage{243}--\blpage{266}
(\byear{1980})
\end{barticle}
\endbibitem

\bibitem[\protect\citeauthoryear{Merz
  et~al.}{2022}]{merzUnderstandingHeavyTails2022}
\begin{barticle}
\bauthor{\bsnm{Merz}, \binits{B.}},
\bauthor{\bsnm{Basso}, \binits{S.}},
\bauthor{\bsnm{Fischer}, \binits{S.}},
\bauthor{\bsnm{Lun}, \binits{D.}},
\bauthor{\bsnm{Bl{\"o}schl}, \binits{G.}},
\bauthor{\bsnm{Merz}, \binits{R.}},
\bauthor{\bsnm{Guse}, \binits{B.}},
\bauthor{\bsnm{Viglione}, \binits{A.}},
\bauthor{\bsnm{Vorogushyn}, \binits{S.}},
\bauthor{\bsnm{Macdonald}, \binits{E.}},
\bauthor{\bsnm{Wietzke}, \binits{L.}},
\bauthor{\bsnm{Schumann}, \binits{A.}}:
\batitle{Understanding {{Heavy Tails}} of {{Flood Peak Distributions}}}.
\bjtitle{Water Resources Research}
\bvolume{58}(\bissue{6}),
\bfpage{2021}--\blpage{030506}
(\byear{2022})
\doiurl{10.1029/2021WR030506}
\end{barticle}
\endbibitem

\bibitem[\protect\citeauthoryear{Naudts}{2008}]{naudtsGeneralisedExponentialFamilies2008}
\begin{barticle}
\bauthor{\bsnm{Naudts}, \binits{J.}}:
\batitle{Generalised exponential families and associated entropy functions}.
\bjtitle{Entropy}
\bvolume{10}(\bissue{3}),
\bfpage{131}--\blpage{149}
(\byear{2008})
\doiurl{10.3390/e10030131}
\end{barticle}
\endbibitem

\bibitem[\protect\citeauthoryear{Nelson}{2022}]{nelsonIndependentApproximatesEnable2022}
\begin{barticle}
\bauthor{\bsnm{Nelson}, \binits{K.P.}}:
\batitle{Independent {{Approximates}} enable closed-form estimation of
  heavy-tailed distributions}.
\bjtitle{Physica A: Statistical Mechanics and its Applications}
\bvolume{601},
\bfpage{127574}
(\byear{2022})
\doiurl{10.1016/j.physa.2022.127574}
\end{barticle}
\endbibitem

\bibitem[\protect\citeauthoryear{Nelson}{2024}]{nelsonOpenProblemsNonextensive2024}
\begin{barticle}
\bauthor{\bsnm{Nelson}, \binits{K.P.}}:
\batitle{Open {{Problems}} within {{Nonextensive Statistical Mechanics}}}.
\bjtitle{Entropy}
\bvolume{26}(\bissue{2}),
\bfpage{118}
(\byear{2024})
\doiurl{10.3390/e26020118}
\end{barticle}
\endbibitem

\bibitem[\protect\citeauthoryear{Newman and
  Sneppen}{1996}]{newmanAvalanchesScalingCoherent1996}
\begin{barticle}
\bauthor{\bsnm{Newman}, \binits{M.E.J.}},
\bauthor{\bsnm{Sneppen}, \binits{K.}}:
\batitle{Avalanches, scaling, and coherent noise}.
\bjtitle{Physical Review E}
\bvolume{54}(\bissue{6}),
\bfpage{6226}--\blpage{6231}
(\byear{1996})
\doiurl{10.1103/PhysRevE.54.6226}
\end{barticle}
\endbibitem

\bibitem[\protect\citeauthoryear{Nelson and
  Thistleton}{2021}]{nelsonCommentsGeneralizedBoxMuller2021}
\begin{barticle}
\bauthor{\bsnm{Nelson}, \binits{K.P.}},
\bauthor{\bsnm{Thistleton}, \binits{W.J.}}:
\batitle{Comments on ``{{Generalized Box-M{\"u}ller Method}} for {{Generating}}
  q-{{Gaussian Random Deviates}}''}.
\bjtitle{IEEE Transactions on Information Theory}
\bvolume{67}(\bissue{10}),
\bfpage{6785}--\blpage{6789}
(\byear{2021})
\doiurl{10.1109/TIT.2021.3071489}
\end{barticle}
\endbibitem

\bibitem[\protect\citeauthoryear{Nelson and
  Umarov}{2010}]{nelsonNonlinearStatisticalCoupling2010}
\begin{barticle}
\bauthor{\bsnm{Nelson}, \binits{K.P.}},
\bauthor{\bsnm{Umarov}, \binits{S.}}:
\batitle{Nonlinear statistical coupling}.
\bjtitle{Physica A: Statistical Mechanics and its Applications}
\bvolume{389}(\bissue{11}),
\bfpage{2157}--\blpage{2163}
(\byear{2010})
\doiurl{10.1016/j.physa.2010.01.044}
\end{barticle}
\endbibitem

\bibitem[\protect\citeauthoryear{Nelson
  et~al.}{2017}]{nelsonAverageUncertaintySystems2017}
\begin{barticle}
\bauthor{\bsnm{Nelson}, \binits{K.P.}},
\bauthor{\bsnm{Umarov}, \binits{S.R.}},
\bauthor{\bsnm{Kon}, \binits{M.A.}}:
\batitle{On the average uncertainty for systems with nonlinear coupling}.
\bjtitle{Physica A: Statistical Mechanics and its Applications}
\bvolume{468},
\bfpage{30}--\blpage{43}
(\byear{2017})
\doiurl{10.1016/j.physa.2016.09.046}
\end{barticle}
\endbibitem

\bibitem[\protect\citeauthoryear{Pickands}{1975}]{pickands1975statistical}
\begin{barticle}
\bauthor{\bsnm{Pickands}, \binits{J.}}:
\batitle{Statistical inference using extreme order statistics}.
\bjtitle{The Annals of Statistics}
\bvolume{3}(\bissue{1}),
\bfpage{119}--\blpage{131}
(\byear{1975})
\doiurl{10.1214/aos/1176343003}
\end{barticle}
\endbibitem

\bibitem[\protect\citeauthoryear{Resnick}{2007}]{resnickHeavyTailPhenomenaProbabilistic2007}
\begin{bbook}
\bauthor{\bsnm{Resnick}, \binits{S.}}:
\bbtitle{Heavy-{{Tail Phenomena}}: {{Probabilistic}} and {{Statistical
  Modeling}}}.
\bpublisher{Springer},
\blocation{New York, NY}
(\byear{2007}).
\doiurl{10.1007/978-0-387-45024-7}
\end{bbook}
\endbibitem

\bibitem[\protect\citeauthoryear{Sarlis and
  Christopoulos}{2012}]{sarlisPredictabilityCoherentnoiseModel2012}
\begin{barticle}
\bauthor{\bsnm{Sarlis}, \binits{N.V.}},
\bauthor{\bsnm{Christopoulos}, \binits{S.-R.G.}}:
\batitle{Predictability of the coherent-noise model and its applications}.
\bjtitle{Physical Review E}
\bvolume{85}(\bissue{5}),
\bfpage{051136}
(\byear{2012})
\doiurl{10.1103/PhysRevE.85.051136}
\end{barticle}
\endbibitem

\bibitem[\protect\citeauthoryear{Shalizi}{2007}]{shaliziMaximumLikelihoodEstimation2007}
\begin{barticle}
\bauthor{\bsnm{Shalizi}, \binits{C.R.}}:
\batitle{Maximum {{Likelihood Estimation}} for q-{{Exponential}} ({{Tsallis}})
  {{Distributions}}}.
\bjtitle{arXiv:math/0701854 [math.ST]}
(\byear{2007})
\doiurl{10.48550/arXiv.math/0701854}
{\href{https://arxiv.org/abs/math/0701854}{{arXiv:math/0701854}}}.
\bcomment{Comment: 4 pages, 1 figure; accompanying R code available from
  http://bactra.org/research/tsallis-MLE/. V2: Added results on estimation from
  censored data, re-arranged introduction, minor corrections and wording
  changes throughout, updated code}
\end{barticle}
\endbibitem

\bibitem[\protect\citeauthoryear{Sneppen and
  Newman}{1997}]{sneppenCoherentNoiseScale1997}
\begin{barticle}
\bauthor{\bsnm{Sneppen}, \binits{K.}},
\bauthor{\bsnm{Newman}, \binits{M.E.}}:
\batitle{Coherent noise, scale invariance and intermittency in large systems}.
\bjtitle{Physica D: Nonlinear Phenomena}
\bvolume{110}(\bissue{3-4}),
\bfpage{209}--\blpage{222}
(\byear{1997})
\end{barticle}
\endbibitem

\bibitem[\protect\citeauthoryear{Tirnakli and
  Borges}{2016}]{tirnakliStandardMapBoltzmannGibbs2016}
\begin{barticle}
\bauthor{\bsnm{Tirnakli}, \binits{U.}},
\bauthor{\bsnm{Borges}, \binits{E.P.}}:
\batitle{The standard map: {{From Boltzmann-Gibbs}} statistics to {{Tsallis}}
  statistics}.
\bjtitle{Scientific Reports}
\bvolume{6}(\bissue{1}),
\bfpage{23644}
(\byear{2016})
\doiurl{10.1038/srep23644}
\end{barticle}
\endbibitem

\bibitem[\protect\citeauthoryear{Thistleton
  et~al.}{2007}]{thistletonGeneralizedBoxMuller2007}
\begin{barticle}
\bauthor{\bsnm{Thistleton}, \binits{W.J.}},
\bauthor{\bsnm{Marsh}, \binits{J.A.}},
\bauthor{\bsnm{Nelson}, \binits{K.}},
\bauthor{\bsnm{Tsallis}, \binits{C.}}:
\batitle{Generalized {{Box}}--{{M{\"u}ller Method}} for {{Generating}}
  q-{{Gaussian Random Deviates}}}.
\bjtitle{IEEE Transactions on Information Theory}
\bvolume{53}(\bissue{12}),
\bfpage{4805}--\blpage{4810}
(\byear{2007})
\doiurl{10.1109/TIT.2007.909173}
\end{barticle}
\endbibitem

\bibitem[\protect\citeauthoryear{Tsallis}{2017}]{tsallisFoundationsStatisticalMechanics2017}
\begin{barticle}
\bauthor{\bsnm{Tsallis}, \binits{C.}}:
\batitle{On the foundations of statistical mechanics}.
\bjtitle{The European Physical Journal Special Topics}
\bvolume{226}(\bissue{7}),
\bfpage{1433}--\blpage{1443}
(\byear{2017})
\doiurl{10.1140/epjst/e2016-60252-2}
\end{barticle}
\endbibitem

\bibitem[\protect\citeauthoryear{Wilke
  et~al.}{1998}]{wilkeAftershocksCoherentnoiseModels1998}
\begin{barticle}
\bauthor{\bsnm{Wilke}, \binits{C.}},
\bauthor{\bsnm{Altmeyer}, \binits{S.}},
\bauthor{\bsnm{Martinetz}, \binits{T.}}:
\batitle{Aftershocks in coherent-noise models}.
\bjtitle{Physica D: Nonlinear Phenomena}
\bvolume{120}(\bissue{3-4}),
\bfpage{401}--\blpage{417}
(\byear{1998})
\end{barticle}
\endbibitem

\bibitem[\protect\citeauthoryear{Wilk and W{\l}odarczyk}{2000}]{Wilk2000}
\begin{barticle}
\bauthor{\bsnm{Wilk}, \binits{G.}},
\bauthor{\bsnm{W{\l}odarczyk}, \binits{Z.}}:
\batitle{Interpretation of the {{Nonextensivity Parameter}} q in {{Some
  Applications}} of {{Tsallis Statistics}} and {{L{\'e}vy Distributions}}}.
\bjtitle{Physical Review Letters}
\bvolume{84}(\bissue{13}),
\bfpage{2770}--\blpage{2773}
(\byear{2000})
\doiurl{10.1103/PhysRevLett.84.2770}
\end{barticle}
\endbibitem

\bibitem[\protect\citeauthoryear{Zaslavsky}{2005}]{zaslavskyHamiltonianChaosFractional2005}
\begin{bbook}
\bauthor{\bsnm{Zaslavsky}, \binits{G.M.}}:
\bbtitle{Hamiltonian Chaos and Fractional Dynamics}.
\bpublisher{Oxford University Press},
\blocation{New York, NY, USA}
(\byear{2005})
\end{bbook}
\endbibitem

\bibitem[\protect\citeauthoryear{Zubillaga
  et~al.}{2025}]{zubillagaThreestateOpinionDynamics2025}
\begin{botherref}
\oauthor{\bsnm{Zubillaga}, \binits{B.J.}},
\oauthor{\bsnm{Granha}, \binits{M.F.B.}},
\oauthor{\bsnm{Wang}, \binits{C.}},
\oauthor{\bsnm{Nelson}, \binits{K.P.}},
\oauthor{\bsnm{Vilela}, \binits{A.L.M.}}:
Three-state opinion dynamics for financial markets on complex networks.
Physica A: Statistical Mechanics and its Applications,
130671
(2025)
\doiurl{10.1016/j.physa.2025.130671}
\end{botherref}
\endbibitem

\bibitem[\protect\citeauthoryear{Çadırcı}{2025}]{cadirci2025nonparametric}
\begin{botherref}
\oauthor{\bsnm{Çadırcı}, \binits{M.S.}}:
Non-parametric goodness-of-fit tests using tsallis entropy measures.
arXiv preprint arXiv:2506.14242
(2025)
\end{botherref}
\endbibitem

\end{thebibliography}

\begin{appendices}

\newpage
\section{Comparison of the coupled and $q$-distributions}\label{secA1}

The field of nonextensive statistical mechanics grew out of analysis of complex systems defined by an escort probability $p_i^{(q)}\equiv\frac{p_i^q}{\sum_{j=1}^{N} p_j^q}$; however, the choice of $q$ as a defining parameter created difficulties in a) relating results to established principles within the statistical analysis of scale-shape distributions, and b) explaining physical theories of complex systems, whose principle property is nonlinear dynamics. Recasting results in nonextensive statistical mechanics, such as this contribution regarding estimation using Independent Approximates, has the potential to integrate advances in modeling complex systems into more establish approaches of statistical analysis and to clarify physical implications. In this appendix, we show that the scale-shape definition of the coupled distributions provides clear mathematical properties which are obscured when using the $\beta$-$q$ translation.  
\begin{table}[!ht]
\caption{Comparison of mathematical properties of the Coupled Gaussian and $q$-Gaussian representations. The results are for the heavy-tailed domain in which $0<\kappa<\infty$ and $1<q<3$.}
\label{dist-prop}
\begin{tabular}{c|c|c|c}
\hline
Property  & \multicolumn{1}{c|}{\begin{tabular}[c]{@{}c@{}}Mathematical Description\end{tabular}} & \multicolumn{1}{c|}{Coupled Gaussian} & \multicolumn{1}{c}{$q$-Gaussian} \\ \hline
\begin{tabular}[c]{@{}l@{}}Inflection Point\end{tabular}   & $f''(x)=0$ & $x=\dfrac{\pm\sigma}{\sqrt{1+2\kappa}}$ & $x=\dfrac{\pm1}{\sqrt{\beta(q+1)}}$  \\ \hline
\begin{tabular}[c]{@{}c@{}}Inflection of \\ Derivative \end{tabular} & $f^{(3)}(x)=0$ & $x=\dfrac{\pm\sigma\sqrt{3}}{\sqrt{1+2\kappa}}$ & $x=\dfrac{\pm\sqrt{3}}{\sqrt{\beta(q+1)}}$ \\ \hline
\begin{tabular}[c]{@{}c@{}}Half asymptotic slope\\ of Log-Log plot\end{tabular} & 
\begin{tabular}[c]{@{}c@{}}$x=e^u$ \\ $g'(u)=\dfrac{e^uf'(e^u)}{f(e^u)}=\dfrac{1}{2} \lim\limits_{u\to\infty} g'(u)$\end{tabular} &  $x=\dfrac{\sigma}{\sqrt{\kappa}}$ & $x=\sqrt{\dfrac{1}{\beta(q-1)}}$  \\ \hline
\begin{tabular}[c]{@{}c@{}}Log-Log Derivative\\ is -1\end{tabular} &  \begin{tabular}[c]{@{}c@{}}$x=e^u$ \\ $g'(u)=-1$\end{tabular}  &   $x=\sigma$  &   $x=\sqrt{\dfrac{1}{\beta(3-q)}}$  \\ \hline
\end{tabular}
\end{table}

Table \ref{dist-prop} shows a comparison of basic mathematical properties of the coupled Gaussian and $q$-Gaussian distributions.  The inflection point of the pdf and its derivative have comparible complexity with the coupled and $q$-Gaussian representations. However, key points of the Log-Log plot of the pdfs show a simplification for the coupled Gaussian.  The point at which the slope of the log-log plot equals $-1$ is always the scale of the distribution, $x=\sigma$. The translation to the $q$-Gaussian does not have this clarity, $x=\sqrt{\dfrac{1}{\beta(3-q)}}$. Likewise, the point at which the log-log slope is half the slope at the asymptotic limit is simply, $|x|=\dfrac{\sigma}{\sqrt{\kappa}}$; whereas, the $q-$Gaussian has the $-1$ constant which complicates interpretation, $|x|=\sqrt{\dfrac{1}{\beta(q-1)}}$.

The significance of the scale for properties of nonadditive entropy was introduced in \cite{nelsonAverageUncertaintySystems2017} and discussed further in \cite{nelsonOpenProblemsNonextensive2024}. For the coupled distributions the density at the scale is equal to an average density defined by the translation of the coupled entropy from the log-density domain back to the density domain.  That is, for the coupled exponential ($\alpha=1$) and the coupled Gaussian ($\alpha=2$) (\ref{CD-PDF}) the density at the scale is equal to the following generalized mean of the distribution
\begin{equation}
    f(\sigma;\kappa,\alpha)=\left(\int_{x\in X} f(x;\kappa,\alpha)^{1+\frac{\alpha\kappa}{1+\kappa}} \,dx \right)^\frac{1+\kappa}{\alpha\kappa}.
\end{equation}
Given the mathematical significance of the coupled distributions at the scale, the physical interpretation of heavy-tailed phenomena should also simplify with this representation. 

The difficulty in interpreting $q$-statistics is illustrated by the interpretations of superstatistics by Beck and Cohen \cite{beckSuperstatistics2003}, and Wilk and Włodarczyk \cite{Wilk2000}. Solving for a generalization of the Boltzmann factor and then normalizing the solution, these investigators concluded that a random variable with a fluctuating standard deviation can be modeled as $q$-exponential distribution with $q$ proportional to the relative variance. The $q$-exponential distribution is exact if the variations $\beta$ are distributed as a gamma distribution. Furthermore, via Taylor series analysis this result is shown to be universal for small fluctuations regardless of the large-scale distribution of the fluctuations. 

However, examination of the result reveals a couple of problems.  First, the relative variance has a domain from 0 to infinity while the heavy-tailed $q$-exponentials can only be normalized from $1<q<2$. Secondly, the superstatistics derivation utilized the Boltzmann factor $e^{-\beta E}$ which neglects the normalization. These issues led to the definition of Type B superstistics, in which the variation in the normalization is included. While the result is still a $q$ or coupled exponential distribution, the relative variance is now equal to the coupling $\kappa$, which like the relative variance has domain over the positive reals for heavy-tailed distributions. 

Anticipating the coupled exponential distribution solution we use $\sigma'$ as the variable scale and the following parameters for the mean and relative variance of the inverse scale:

\begin{equation}
    \frac{1}{\sigma} = \left< \frac{1}{\sigma'}\right>, \quad
  \kappa = \frac{\left< \frac{1}{\sigma'^2}\right> - \left< \frac{1}{\sigma'}\right>^2}{\left< \frac{1}{\sigma'}\right>^2}.
\end{equation}

With the normalization treated as a constant, the Type A superstatistics result is
\begin{equation}
    C \left( 1 + \kappa \frac{x}{\sigma}\right)^{-\frac{1}{\kappa}} = C \int_{0}^{\infty} e^{-\frac{x}{\sigma'}}\frac{ \left( \sigma \kappa^{-1}\right)^\frac{1}{\kappa}}{\Gamma \left(\frac{1}{\kappa}\right)} \left(\frac{1}{\sigma'}\right)^{\frac{1}{\kappa}-1} e^{-\frac{\sigma}{\kappa}\frac{1}{\sigma'}} \,d\left(\frac{1}{\sigma'}\right).
\end{equation}
From this result, the $q$-exponential distribution is formed by the substitutions,
\begin{equation}
    \kappa=q-1, \quad \sigma = \frac{1}{\beta}, \quad C=\beta (2-q),
\end{equation}
which provided the interpretation that the relative variance is proportional to $q$. In contrast, if the normalization of the exponential distribution is included the Type B superstastics result is
\begin{equation}
    \frac{1}{\sigma} \left( 1 + \kappa \frac{x}{\sigma}\right)^{-\left(\frac{1}{\kappa}+1\right)} =  \int_{0}^{\infty} \frac{1}{\sigma'}e^{-\frac{x}{\sigma'}}\frac{ \left( \sigma \kappa^{-1}\right)^\frac{1}{\kappa}}{\Gamma \left(\frac{1}{\kappa}\right)} \left(\frac{1}{\sigma'}\right)^{\frac{1}{\kappa}-1} e^{-\frac{\sigma}{\kappa}\frac{1}{\sigma'}} \,d\left(\frac{1}{\sigma'}\right).
\end{equation}
Now, the result is precisely the normalized coupled exponential distribution and the relative variance is equal to the coupling.  The translation to $q$ using (\ref{q-coupling}) is $\kappa=\frac{q-1}{2-q}$. While there may be some applications of Type A superstatistics relevant to generalizations of the Boltzmann factor, it cannot be used to derive a distribution in which the variation of the normalization was neglected. 

\section{}\label{secA2}

\begin{table}[!h]
\captionsetup{justification=centering, width=\textwidth}
\caption{Empirical mean square errors (MSE) of parameter estimates for data generated from a Coupled Gaussian distribution with a sample size $n = 1000$. The scale $\sigma=0.5$ and the shape $\kappa$ varies as indicated.}
\begin{tabular}{l|llll}
\hline
\multirow{3}{*}{$\kappa$} & \multicolumn{4}{c}{\begin{tabular}[c]{@{}c@{}}$(MSE\pm SD)$x$10^{-3}$ \end{tabular}} \\ \cmidrule{2-5} 
                       & \multicolumn{2}{c|}{IA\_GM}                                                    & \multicolumn{2}{c}{ML}                                 \\ \cmidrule{2-5} 
                       & \multicolumn{1}{l|}{$\hat{\kappa}$}            & \multicolumn{1}{l|}{$\hat{\sigma}$}     & \multicolumn{1}{l|}{$\hat{\kappa}$}  & $\hat{\sigma}$  \\ \hline
0.25                   & \multicolumn{1}{l|}{$71\pm5$}         & \multicolumn{1}{l|}{$8\pm2$}        & \multicolumn{1}{l|}{$19\pm3$}        & $17\pm3$        \\ \hline
0.5                    & \multicolumn{1}{l|}{$30\pm10$}         & \multicolumn{1}{l|}{$5\pm3$}        & \multicolumn{1}{l|}{$3\pm6$}        & $12\pm3$        \\ \hline
1 & \multicolumn{1}{l|}{$20\pm20$} & \multicolumn{1}{l|}{$8\pm4$}        & \multicolumn{1}{l|}{$10\pm8$}        & $10\pm4$  \\ \hline
1.25 & \multicolumn{1}{l|}{$10\pm20$}  & \multicolumn{1}{l|}{$21\pm5$}    & \multicolumn{1}{l|}{$10\pm10$}  & $11\pm3$        \\ \hline
2   & \multicolumn{1}{l|}{$49\pm25$} & \multicolumn{1}{l|}{$41\pm6$} & \multicolumn{1}{l|}{$20\pm10$}        & $13\pm3$        \\ \hline
\end{tabular}
\end{table}

\begin{table}[!h]
\captionsetup{justification=centering, width=\textwidth}
\caption{Goodness-of-Fit Metrics for the Coupled Gaussian Distribution under Various Methods and Shape Parameters $\kappa$ with a Fixed Scale Parameter ($\sigma = 0.5$) and sample size$=1000$.}
\label{tab:Goodness-t}
\begin{tabular}{llllll}
\hline
\multicolumn{6}{c}{Coupled Gaussian} \\ \hline
\multicolumn{6}{c}{Average deviation (AD)    $\sigma=0.5$} \\ \hline
\multicolumn{1}{l|}{Method\textbackslash $\kappa$} & \multicolumn{1}{l|}{0.25} & \multicolumn{1}{l|}{0.5} & \multicolumn{1}{l|}{1} & \multicolumn{1}{l|}{1.25} & 2 \\ \hline
\multicolumn{1}{l|}{IA (Geometric mean)} & \multicolumn{1}{l|}{0.067} & \multicolumn{1}{l|}{0.12} & \multicolumn{1}{l|}{0.69} & \multicolumn{1}{l|}{1.7} & 163 \\ \hline
\multicolumn{1}{l|}{ML} & \multicolumn{1}{l|}{0.051} & \multicolumn{1}{l|}{0.11} & \multicolumn{1}{l|}{0.78} & \multicolumn{1}{l|}{2.61} & 150\\ \hline
\multicolumn{6}{c}{Cramer–von Mises (CvM)  $\sigma=0.5$} \\ \hline
\multicolumn{1}{l|}{Method\textbackslash $\kappa$} & \multicolumn{1}{l|}{0.25} & \multicolumn{1}{l|}{0.5} & \multicolumn{1}{l|}{1} & \multicolumn{1}{l|}{1.25} & 2 \\ \hline
\multicolumn{1}{l|}{IA (Geometric mean)} & \multicolumn{1}{l|}{0.012} & \multicolumn{1}{l|}{0.023} & \multicolumn{1}{l|}{0.025} & \multicolumn{1}{l|}{0.022} & 0.023 \\ \hline
\multicolumn{1}{l|}{ML} & \multicolumn{1}{l|}{0.13} & \multicolumn{1}{l|}{0.15} & \multicolumn{1}{l|}{0.14} & \multicolumn{1}{l|}{0.13} & 0.093 \\ \hline
\multicolumn{6}{c}{NLL} \\ \hline
\multicolumn{1}{l|}{Method\textbackslash $\kappa$} & \multicolumn{1}{l|}{0.25} & \multicolumn{1}{l|}{0.5} & \multicolumn{1}{l|}{1} & \multicolumn{1}{l|}{1.25} & 2 \\ \hline
\multicolumn{1}{l|}{IA (Geometric mean)} & \multicolumn{1}{l|}{1,900} & \multicolumn{1}{l|}{1,500} & \multicolumn{1}{l|}{1,800} & \multicolumn{1}{l|}{2,100} & 3,800 \\ \hline
\multicolumn{1}{l|}{ML} & \multicolumn{1}{l|}{1,600} & \multicolumn{1}{l|}{1,500} & \multicolumn{1}{l|}{1,800} & \multicolumn{1}{l|}{2,100} & 3,700 \\ \hline
\end{tabular}
\end{table}

\begin{table}[!h]
\captionsetup{justification=centering, width=\textwidth}
\caption{Empirical mean square errors (MSE) of parameter estimates for data generated from a Coupled Gaussian distribution with a sample size $n = 100$. The scale $\sigma=0.5$ and the shape $\kappa$ varies as indicated.}
\begin{tabular}{l|llll}
\hline
\multirow{3}{*}{$\kappa$} & \multicolumn{4}{c}{\begin{tabular}[c]{@{}c@{}}$(MSE\pm SD)$x$10^{-3}$\end{tabular}} \\ \cmidrule{2-5} 
                       & \multicolumn{2}{c|}{IA\_GM}                                                    & \multicolumn{2}{c}{ML}                                 \\ \cmidrule{2-5} 
                       & \multicolumn{1}{l|}{$\hat{\kappa}$}            & \multicolumn{1}{l|}{$\hat{\sigma}$}     & \multicolumn{1}{l|}{$\hat{\kappa}$}  & $\hat{\sigma}$  \\ \hline
0.25                   & \multicolumn{1}{l|}{$53\pm8$}         & \multicolumn{1}{l|}{$81\pm8$}        & \multicolumn{1}{l|}{$10\pm9$}        & $20\pm3$        \\ \hline
0.5                    & \multicolumn{1}{l|}{$20\pm10$}  & \multicolumn{1}{l|}{$60\pm4$}        & \multicolumn{1}{l|}{$20\pm10$}        & $17\pm3$        \\ \hline
1 & \multicolumn{1}{l|}{$216\pm9$} & \multicolumn{1}{l|}{$20\pm10$}        & \multicolumn{1}{l|}{$60\pm20$}        & $13\pm3$  \\ \hline
1.25 & \multicolumn{1}{l|}{$450\pm10$}  & \multicolumn{1}{l|}{$118\pm2$}    & \multicolumn{1}{l|}{$20\pm20$}  & $10\pm2$        \\ \hline
2   & \multicolumn{1}{l|}{$150\pm20$} & \multicolumn{1}{l|}{$14\pm3$} & \multicolumn{1}{l|}{$120\pm20$} & $62\pm3$        \\ \hline
\end{tabular}
\end{table}

\begin{table}[!h]
\captionsetup{justification=centering, width=\textwidth}
\caption{Goodness-of-Fit Metrics for the Coupled Gaussian Distribution under Various Methods and Shape Parameters $\kappa$ with a Fixed Scale Parameter ($\sigma = 0.5$) and sample size$=100$.}
\label{tab:Goodness-t}
\begin{tabular}{llllll}
\hline
\multicolumn{6}{c}{Coupled Gaussian} \\ \hline
\multicolumn{6}{c}{Average deviation (AD)    $\sigma=0.5$} \\ \hline
\multicolumn{1}{l|}{Method\textbackslash $\kappa$} & \multicolumn{1}{l|}{0.25} & \multicolumn{1}{l|}{0.5} & \multicolumn{1}{l|}{1} & \multicolumn{1}{l|}{1.25} & 2 \\ \hline
\multicolumn{1}{l|}{IA (Geometric mean)} & \multicolumn{1}{l|}{0.081} & \multicolumn{1}{l|}{0.19} & \multicolumn{1}{l|}{1.5} & \multicolumn{1}{l|}{2.7} & 160 \\ \hline
\multicolumn{1}{l|}{ML} & \multicolumn{1}{l|}{0.21} & \multicolumn{1}{l|}{0.45} & \multicolumn{1}{l|}{2.6} & \multicolumn{1}{l|}{7.6} & 270\\ \hline
\multicolumn{6}{c}{Cramer–von Mises (CvM)  $\sigma=0.5$} \\ \hline
\multicolumn{1}{l|}{Method\textbackslash $\kappa$} & \multicolumn{1}{l|}{0.25} & \multicolumn{1}{l|}{0.5} & \multicolumn{1}{l|}{1} & \multicolumn{1}{l|}{1.25} & 2 \\ \hline
\multicolumn{1}{l|}{IA (Geometric mean)} & \multicolumn{1}{l|}{0.23} & \multicolumn{1}{l|}{0.22} & \multicolumn{1}{l|}{0.24} & \multicolumn{1}{l|}{0.20} & 0.48 \\ \hline
\multicolumn{1}{l|}{ML} & \multicolumn{1}{l|}{0.31} & \multicolumn{1}{l|}{0.30} & \multicolumn{1}{l|}{0.29} & \multicolumn{1}{l|}{0.26} & 0.38 \\ \hline
\multicolumn{6}{c}{NLL} \\ \hline
\multicolumn{1}{l|}{Method\textbackslash $\kappa$} & \multicolumn{1}{l|}{0.25} & \multicolumn{1}{l|}{0.5} & \multicolumn{1}{l|}{1} & \multicolumn{1}{l|}{1.25} & 2 \\ \hline
\multicolumn{1}{l|}{IA (Geometric mean)} & \multicolumn{1}{l|}{160} & \multicolumn{1}{l|}{150} & \multicolumn{1}{l|}{210} & \multicolumn{1}{l|}{250} & 250 \\ \hline
\multicolumn{1}{l|}{ML} & \multicolumn{1}{l|}{180} & \multicolumn{1}{l|}{160} & \multicolumn{1}{l|}{210} & \multicolumn{1}{l|}{240} & 260 \\ \hline
\end{tabular}
\end{table}

\begin{table}[!h]
\captionsetup{justification=centering, width=\textwidth}
\caption{Empirical mean square errors (MSE) of parameter estimates for data generated from a Coupled Exponential distribution with a sample size $n = 1000$. The scale $\sigma=0.5$ and the shape $\kappa$ varies as indicated.}
\begin{tabular}{l|llllll}
\hline
\multicolumn{1}{c|}{\multirow{3}{*}{$\kappa$}} & \multicolumn{6}{c}{\begin{tabular}[c]{@{}c@{}}$(MSE\pm SD)$x$10^{-3}$\end{tabular}}                                  \\ \cmidrule{2-7} 
\multicolumn{1}{c|}{}                          & \multicolumn{2}{c|}{IA\_GM}                             & \multicolumn{2}{c|}{IA}                                 & \multicolumn{2}{l}{ML}             \\ \cmidrule{2-7} 
\multicolumn{1}{c|}{}                          & \multicolumn{1}{l|}{$\hat{\kappa}$} & \multicolumn{1}{l|}{$\hat{\sigma}$} & \multicolumn{1}{l|}{$\hat{\kappa}$} & \multicolumn{1}{l|}{$\hat{\sigma}$} & \multicolumn{1}{l|}{$\hat{\kappa}$} & $\hat{\sigma}$ \\ \hline
0.25& \multicolumn{1}{l|}{6±2} & \multicolumn{1}{l|}{1.0±0.1} & \multicolumn{1}{l|}{0±4} & \multicolumn{1}{l|}{1±3} & \multicolumn{1}{l|}{22±4} & 5±4 \\ \hline
0.5& \multicolumn{1}{l|}{117±5} & \multicolumn{1}{l|}{34.0±0.3} & \multicolumn{1}{l|}{30±3} & \multicolumn{1}{l|}{33±5} & \multicolumn{1}{l|}{20±10} & 4±3 \\ \hline
1& \multicolumn{1}{l|}{1±2}  & \multicolumn{1}{l|}{2.0±0.2}& \multicolumn{1}{l|}{27±20} & \multicolumn{1}{l|}{68±3} & \multicolumn{1}{l|}{20±10} & 4±5 \\ \hline
1.25& \multicolumn{1}{l|}{30±10} & \multicolumn{1}{l|}{7.0±0.1} & \multicolumn{1}{l|}{9±10} & \multicolumn{1}{l|}{13±4} & \multicolumn{1}{l|}{30±10} & 3±2 \\ \hline
2& \multicolumn{1}{l|}{80±10} & \multicolumn{1}{l|}{24.0±0.2} & \multicolumn{1}{l|}{10±20} & \multicolumn{1}{l|}{50±10} & \multicolumn{1}{l|}{30±10} & 2±2 \\ \hline
\end{tabular}
\end{table}

\begin{table}[!h]
\captionsetup{justification=centering, width=\textwidth}
\caption{Goodness-of-Fit Metrics for the Coupled Exponential Distribution under Various Methods and Shape Parameters $\kappa$ with a Fixed Scale Parameter ($\sigma = 0.5$) and sample size=$1000$.}
\label{tab:Goodness-GP}
\begin{tabular}{llllll}
\hline
\multicolumn{6}{c}{Coupled Exponential} \\ \hline
\multicolumn{6}{c}{Average deviation (AD)    $\sigma=0.5$} \\ \hline
\multicolumn{1}{l|}{Method\textbackslash $\kappa$} & \multicolumn{1}{l|}{0.25} & \multicolumn{1}{l|}{0.5} & \multicolumn{1}{l|}{1} & \multicolumn{1}{l|}{1.25} & 2 \\ \hline
\multicolumn{1}{l|}{IA (Geometric mean)} & \multicolumn{1}{l|}{0.004} & \multicolumn{1}{l|}{0.12} & \multicolumn{1}{l|}{0.021} & \multicolumn{1}{l|}{0.79} & 140 \\ \hline
\multicolumn{1}{l|}{IA (Triplets)} & \multicolumn{1}{l|}{0.002} & \multicolumn{1}{l|}{0.033} & \multicolumn{1}{l|}{0.17} & \multicolumn{1}{l|}{0.16} & 4.1 \\ \hline
\multicolumn{1}{l|}{ML} & \multicolumn{1}{l|}{0.011} & \multicolumn{1}{l|}{0.023} & \multicolumn{1}{l|}{0.17} & \multicolumn{1}{l|}{0.53} & 29 \\ \hline
\multicolumn{6}{c}{Cramer–von Mises (CvM)  $\sigma=0.5$} \\ \hline
\multicolumn{1}{l|}{Method\textbackslash $\kappa$} & \multicolumn{1}{l|}{0.25} & \multicolumn{1}{l|}{0.5} & \multicolumn{1}{l|}{1} & \multicolumn{1}{l|}{1.25} & 2 \\ \hline
\multicolumn{1}{l|}{IA (Geometric mean)} & \multicolumn{1}{l|}{0.0004} & \multicolumn{1}{l|}{0.066} & \multicolumn{1}{l|}{0.0005} & \multicolumn{1}{l|}{0.002} & 0.016 \\ \hline
\multicolumn{1}{l|}{IA (Triplets)} & \multicolumn{1}{l|}{0.0004} & \multicolumn{1}{l|}{0.13} & \multicolumn{1}{l|}{0.45} & \multicolumn{1}{l|}{0.014} & 0.13 \\ \hline
\multicolumn{1}{l|}{ML} & \multicolumn{1}{l|}{0.003} & \multicolumn{1}{l|}{0.002} & \multicolumn{1}{l|}{0.002} & \multicolumn{1}{l|}{0.002} & 0.003 \\ \hline
\multicolumn{6}{c}{NLL} \\ \hline
\multicolumn{1}{l|}{Method\textbackslash $\kappa$} & \multicolumn{1}{l|}{0.25} & \multicolumn{1}{l|}{0.5} & \multicolumn{1}{l|}{1} & \multicolumn{1}{l|}{1.25} & 2 \\ \hline
\multicolumn{1}{l|}{IA (Geometric mean)} & \multicolumn{1}{l|}{1100} & \multicolumn{1}{l|}{1610} & \multicolumn{1}{l|}{2600} & \multicolumn{1}{l|}{3100} & 4600 \\ \hline
\multicolumn{1}{l|}{IA (Triplets)} & \multicolumn{1}{l|}{1100} & \multicolumn{1}{l|}{1600} & \multicolumn{1}{l|}{2600} & \multicolumn{1}{l|}{3100} & 4600 \\ \hline
\multicolumn{1}{l|}{ML} & \multicolumn{1}{l|}{1100} & \multicolumn{1}{l|}{1600} & \multicolumn{1}{l|}{2600} & \multicolumn{1}{l|}{3100} & 4600 \\ \hline
\end{tabular}
\end{table}

\begin{table}[!h]
\captionsetup{justification=centering, width=\textwidth}
\caption{Empirical mean square errors (MSE) of parameter estimates for data generated from a Coupled Exponential distribution with a sample size $n = 100$. The scale $\sigma=0.5$ and the shape $\kappa$ varies as indicated.}
\begin{tabular}{l|llllll}
\hline
\multicolumn{1}{c|}{\multirow{3}{*}{$\kappa$}} & \multicolumn{6}{c}{\begin{tabular}[c]{@{}c@{}}$(MSE\pm SD)$x$10^{-3}$\end{tabular}}                                  \\ \cmidrule{2-7} 
\multicolumn{1}{c|}{}                          & \multicolumn{2}{c|}{IA\_GM}                             & \multicolumn{2}{c|}{IA}                                 & \multicolumn{2}{l}{ML}             \\ \cmidrule{2-7} 
\multicolumn{1}{c|}{}                          & \multicolumn{1}{l|}{$\hat{\kappa}$} & \multicolumn{1}{l|}{$\hat{\sigma}$} & \multicolumn{1}{l|}{$\hat{\kappa}$} & \multicolumn{1}{l|}{$\hat{\sigma}$} & \multicolumn{1}{l|}{$\hat{\kappa}$} & $\hat{\sigma}$ \\ \hline
0.25& \multicolumn{1}{l|}{6±3} & \multicolumn{1}{l|}{10±1} & \multicolumn{1}{l|}{114±5} & \multicolumn{1}{l|}{77±3} & \multicolumn{1}{l|}{113±5} & 20±5 \\ \hline
0.5& \multicolumn{1}{l|}{125±6} & \multicolumn{1}{l|}{20±1} & \multicolumn{1}{l|}{250±10} & \multicolumn{1}{l|}{253±4} & \multicolumn{1}{l|}{121±6} & 19±4 \\ \hline
1& \multicolumn{1}{l|}{40±30}  & \multicolumn{1}{l|}{10±1}& \multicolumn{1}{l|}{50±20} & \multicolumn{1}{l|}{88±2} & \multicolumn{1}{l|}{140±10} & 16±5 \\ \hline
1.25& \multicolumn{1}{l|}{0±10} & \multicolumn{1}{l|}{32±2} & \multicolumn{1}{l|}{410±10} & \multicolumn{1}{l|}{264±3} & \multicolumn{1}{l|}{10±20} & 13±2 \\ \hline
2& \multicolumn{1}{l|}{10±20} & \multicolumn{1}{l|}{83±3} & \multicolumn{1}{l|}{30±10} & \multicolumn{1}{l|}{157±5} & \multicolumn{1}{l|}{200±20} & 27±2\\ \hline
\end{tabular}
\end{table}

\begin{table}[!h]
\captionsetup{justification=centering, width=\textwidth}
\caption{Goodness-of-Fit Metrics for the Coupled Exponential Distribution under Various Methods and Shape Parameters $\kappa$ with a Fixed Scale Parameter ($\sigma = 0.5$) and sample size=$100$}
\label{tab:Goodness-GP}
\begin{tabular}{llllll}
\hline
\multicolumn{6}{c}{Coupled Exponential} \\ \hline
\multicolumn{6}{c}{Average deviation (AD)    $\sigma=0.5$} \\ \hline
\multicolumn{1}{l|}{Method\textbackslash $\kappa$} & \multicolumn{1}{l|}{0.25} & \multicolumn{1}{l|}{0.5} & \multicolumn{1}{l|}{1} & \multicolumn{1}{l|}{1.25} & 2 \\ \hline
\multicolumn{1}{l|}{IA (Geometric mean)} & \multicolumn{1}{l|}{0.016} & \multicolumn{1}{l|}{0.089} & \multicolumn{1}{l|}{0.18} & \multicolumn{1}{l|}{0.16} & 4.0 \\ \hline
\multicolumn{1}{l|}{IA (Triplets)} & \multicolumn{1}{l|}{0.040} & \multicolumn{1}{l|}{0.16} & \multicolumn{1}{l|}{0.13} & \multicolumn{1}{l|}{1.3} & 5.5 \\ \hline
\multicolumn{1}{l|}{ML} & \multicolumn{1}{l|}{0.046} & \multicolumn{1}{l|}{0.087} & \multicolumn{1}{l|}{0.35} & \multicolumn{1}{l|}{0.70} & 8.6 \\ \hline
\multicolumn{6}{c}{Cramer–von Mises (CvM)  $\sigma=0.5$} \\ \hline
\multicolumn{1}{l|}{Method\textbackslash $\kappa$} & \multicolumn{1}{l|}{0.25} & \multicolumn{1}{l|}{0.5} & \multicolumn{1}{l|}{1} & \multicolumn{1}{l|}{1.25} & 2 \\ \hline
\multicolumn{1}{l|}{IA (Geometric mean)} & \multicolumn{1}{l|}{0.003} & \multicolumn{1}{l|}{0.036} & \multicolumn{1}{l|}{0.006} & \multicolumn{1}{l|}{0.008} & 0.018 \\ \hline
\multicolumn{1}{l|}{IA (Triplets)} & \multicolumn{1}{l|}{0.031} & \multicolumn{1}{l|}{0.24} & \multicolumn{1}{l|}{0.036} & \multicolumn{1}{l|}{0.43} & 0.074 \\ \hline
\multicolumn{1}{l|}{ML} & \multicolumn{1}{l|}{0.005} & \multicolumn{1}{l|}{0.006} & \multicolumn{1}{l|}{0.006} & \multicolumn{1}{l|}{0.007} & 0.15 \\ \hline
\multicolumn{6}{c}{NLL} \\ \hline
\multicolumn{1}{l|}{Method\textbackslash $\kappa$} & \multicolumn{1}{l|}{0.25} & \multicolumn{1}{l|}{0.5} & \multicolumn{1}{l|}{1} & \multicolumn{1}{l|}{1.25} & 2 \\ \hline
\multicolumn{1}{l|}{IA (Geometric mean)} & \multicolumn{1}{l|}{49} & \multicolumn{1}{l|}{73} & \multicolumn{1}{l|}{120} & \multicolumn{1}{l|}{140} & 210 \\ \hline
\multicolumn{1}{l|}{IA (Triplets)} & \multicolumn{1}{l|}{49} & \multicolumn{1}{l|}{75} & \multicolumn{1}{l|}{120} & \multicolumn{1}{l|}{150} & 150 \\ \hline
\multicolumn{1}{l|}{ML} & \multicolumn{1}{l|}{49} & \multicolumn{1}{l|}{72} & \multicolumn{1}{l|}{120} & \multicolumn{1}{l|}{140} & 150 \\ \hline
\end{tabular}
\end{table}


\end{appendices}



\end{document}